\documentclass[pdflatex,sn-mathphys-num,]{sn-jnl}

\usepackage{graphicx}%
\usepackage{multirow}%
\usepackage{amsmath,amssymb,amsfonts,physics}%
\usepackage{amsthm}%
\usepackage{mathrsfs}%
\usepackage[title]{appendix}%
\usepackage{xcolor}%
\usepackage{textcomp}%
\usepackage{manyfoot}%
\usepackage{booktabs}%
\usepackage{algorithm}%
\usepackage{algorithmicx}%
\usepackage{algpseudocode}%
\usepackage{listings}%
\usepackage{siunitx}
\usepackage{ulem}

\newcommand{\che}{{\footnotesize CHECK} }
\newcommand{\est}{{\footnotesize ESTIMATION} }
\newcommand{\ave}[1]{\langle#1\rangle}

\newcommand{\cP}{{\mathcal{P}}}

\newcommand{\X}{{\mathsf{X}}}

\newcommand{\ii}{{\rm i}}
\makeatletter
\let\X\@undefined
\makeatother
\newcommand{\X}{{\mathcal{X}}}

\renewcommand{\tr}[1]{{\rm Tr}[#1]}

\newcommand{\mse}[1]{{\rm MSE}_{#1}}

\theoremstyle{thmstyleone}%
\newtheorem{theorem}{Theorem}

\theoremstyle{thmstyletwo}%

\theoremstyle{thmstylethree}%

\begin{document}

\title[Article Title]{Certification of Network Quantum Sensing}

\author*[1]{\fnm{Matteo} \sur{Rosati}}\email{matteo.rosati@uniroma3.it}

\author[2]{\fnm{Gabriele} \sur{Bizzarri}}\email{gabriele.bizzarri@uniroma3.it}

\author[2,3,4]{\fnm{Marco} \sur{Barbieri}}\email{marco.barbieri@uniroma3.it}

\affil*[1]{\orgdiv{Dipartimento di Ingegneria Civile, Informatica e delle Tecnologie Aeronautiche}, \orgname{Universit\`a degli Studi Roma Tre}, \orgaddress{\street{Via della Vasca Navale 79}, \city{Rome}, \postcode{00146}, \country{Italy}}}

\affil[2]{\orgdiv{Dipartimento di Scienze}, \orgname{Universit\`a degli Studi Roma Tre}, \orgaddress{\street{Via della Vasca Navale 84}, \city{Rome}, \postcode{00146}, \country{Italy}}}

\affil[3]{\orgdiv{Istituto Nazionale di Ottica}, \orgname{CNR}, \orgaddress{\street{Largo E. Fermi 6}, \city{Firenze}, \postcode{50125}, \country{Italy}}}

\affil[4]{\orgdiv{INFN}, \orgname{Sezione di Roma Tre}, \orgaddress{\street{Via della Vasca Navale 84}, \city{Rome}, \postcode{00146}, \country{Italy}}}

\abstract{The distribution of quantum sensors on quantum networks is a key enabler of quantum technologies in interferometry, gravimetry, timekeeping, biological monitoring, and beyond. Yet, guaranteeing the security of these distributed sensors over noisy, insecure networks remains a formidable challenge. Previous efforts to combine quantum metrology and cryptography have encountered an apparently unavoidable tension, proposing bounds for security which are only loosely tied to the achievable measurement performance.
Here we introduce a quantum remote sensing protocol that can rigorously certify privacy and integrity of the estimation. By employing offline bilateral Pauli-twirling, our approach forces the effective quantum channel into a Bell-diagonal form, independently of the attack. Surprisingly, this also preserves metrological sensitivity without introducing additional experimental overhead. Relying solely on public communication alongside an insecure quantum link, the protocol enables legitimate users to exactly quantify their estimation error relative to an eavesdropper controlling the channels. 
 We experimentally demonstrate this framework by estimating an optical phase using entangled photons, observing that the users' precision consistently surpasses the eavesdropper's capabilities across a broad parameter regime. By unifying quantum cryptography and metrology, our results provide a practical pathway to achieve simultaneous quantum-limited precision and rigorous information security in real-world quantum networks.}

\maketitle

\section{Introduction}\label{sec1}

The transition of the quantum internet from idea to reality promises a profound shift in how quantum information is treated: a unified network where communication, delegated computation, and distributed sensing seamlessly intertwine~\cite{Kimble2008,Zhang2026,Ge2025,Kim2024,Liu2024d,Liu2024a,Dalvit2024,Malia2022a,Zhao2021d,Guo2020,Gessner2020,Proctor2018,Hassani2024,Alushi2025,Rosati2022,Rosati2021,Barz2013,Polacchi2023,Wei2025}, as routinely happens in modern classical networks~\cite{Dong2023,Chen2025a,Tao2026}. Among these capabilities, distributed quantum metrology is emerging as a transformative frontier: by coordinating spatially separated entangled nodes, networked quantum sensors unlock qualitatively new capabilities.
for interferometry, gravimetry, timekeeping, energy technology, and biological monitoring~\cite{Stas2026,Novikov2025,Crawford2025,Aslam2023,Stray2022,Komar2013}. Yet, as quantum networks scale toward real-world deployment, a critical imperative emerges: protecting the integrity and privacy of the estimated data. Integrating high-precision remote sensing with quantum-safe, physical-layer cryptography is essential for monitoring sensitive quantities, such as strategic geospatial intelligence and private biometric information. Unfortunately, achieving this integration represents far more than an engineering hurdle; it exposes a fundamental physical tension between two opposing operational paradigms. 

Quantum metrology is fundamentally a science of characterization: it demands comprehensive modelling of the physical system and alignment of probes and measurements to maximize sensitivity to target parameters~\cite{Giovannetti2011,Pirandola2018a,Marciniak2022}. Cryptography, conversely, is a science of obfuscation: to preclude an eavesdropper from exploiting physical or logical trapdoors, cryptographic protocols deliberately minimize trust in the system and routinely employ randomization~\cite{Portmann2022,Pirandola2020,Renner2007,Yin2020a,Paraiso2021a,Nadlinger2022}. This creates an apparent physical paradox: naively, it seems that a quantum probe cannot be simultaneously aligned, in order to precisely estimate a parameter, and isotropically randomized, in order to remain opaque to malicious interference.

Previous efforts in sensing with private quantum resources (SPQR) have grappled with this dichotomy, implying that one must trade metrological precision and efficiency for cryptographic security. 
Existing proposals have either weakened security guarantees by making explicit, exploitable assumptions about the communication channel~\cite{Kianvash2025a,Moore2023}, or severely underestimated the achievable precision by bounding performance against worst-case deviations~\cite{Bizzarri2025a,Moore2025,Ho2024,Huang2019a}, or introduced demanding resource overheads via full state tomography and trap codes inspired by quantum error correction~\cite{Yin2019,Shettell2022a}.
Collectively, these approaches encapsulate the prevailing perception that high security and precision are inherently competing resources in the presence of an uncharacterized adversary.

In this work, we demonstrate that this tension is not fundamental. We present an SPQR protocol that definitively certifies both optimal estimation precision and robust security, free from the compromises of previous proposals. We resolve the cryptography-metrology conflict through a rather counterintuitive application of Pauli twirling. While twirling is a cornerstone randomization technique in quantum cryptography~\cite{Shor2000,Kraus2005}—used to enforce a predictable, isotropic channel independently of an eavesdropper's actions, we reveal that this global scrambling does not destroy the metrological utility of the system, despite its non-commutation with the parameter dynamics' generator. Namely, by means of a virtual twirling applied to the entangled probe shared between the sender and the parameter encoder, we demonstrate that isotropic cryptographic obfuscation can robustly coexist with the phase-sensitive alignment required for quantum metrology. 

We ground our theoretical framework by applying our protocol to a remote phase-sensing task, providing both rigorous analytical guarantees and an experimental validation. By utilizing Pauli twirling to enforce a strictly characterizable channel, we unlock the ability to analytically compute the exact Cram\'er-Rao bounds for both the legitimate users and the eavesdropper, entirely free from assumptions regarding the adversary's attack strategy.  Crucially, we mathematically prove, through the analysis of sheared covariant parameter families, that the eavesdropper's precision is fundamentally phase-independent, ensuring a fixed security floor. We subsequently validate these theoretical guarantees via a proof-of-principle photonic experiment. Deliberately employing a minimal-resource setup to showcase the protocol's accessibility, e.g., twirling can be implemented in post-processing via a classical frame-update, we observe remarkable agreement with our theoretical predictions. The experiment confirms that the legitimate users maintain a commanding precision advantage over the eavesdropper across a broad parameter regime, directly demonstrating the protocol's practical viability. Our results establish a rigorous and unexpected compatibility between security and sensing, providing a certification blueprint for the future of distributed sensing in secure quantum networks.

\section{Results}\label{sec2}
\subsection{Remote sensing with Pauli-twirled probes}
We consider a quantum network enabling end-to-end connection of a sender, Alice, and a receiver, Bob, as illustrated in Fig.~\ref{fig:concept}. Alice is a central hub with full quantum capabilities, able to prepare, distribute and measure a bipartite entangled state. Bob, instead, is a peripheral node with his quantum capabilities limited to single-party measurements. The objective of remote sensing over the network is for Alice to estimate a parameter $\phi$ precisely, but, crucially, only Bob has local access to the parameter-encoding process she wishes to characterise. He is thus charged with querying the process by using the received quantum state and performing measurements. 

\begin{figure*}
    \centering
    \includegraphics[width=\linewidth]{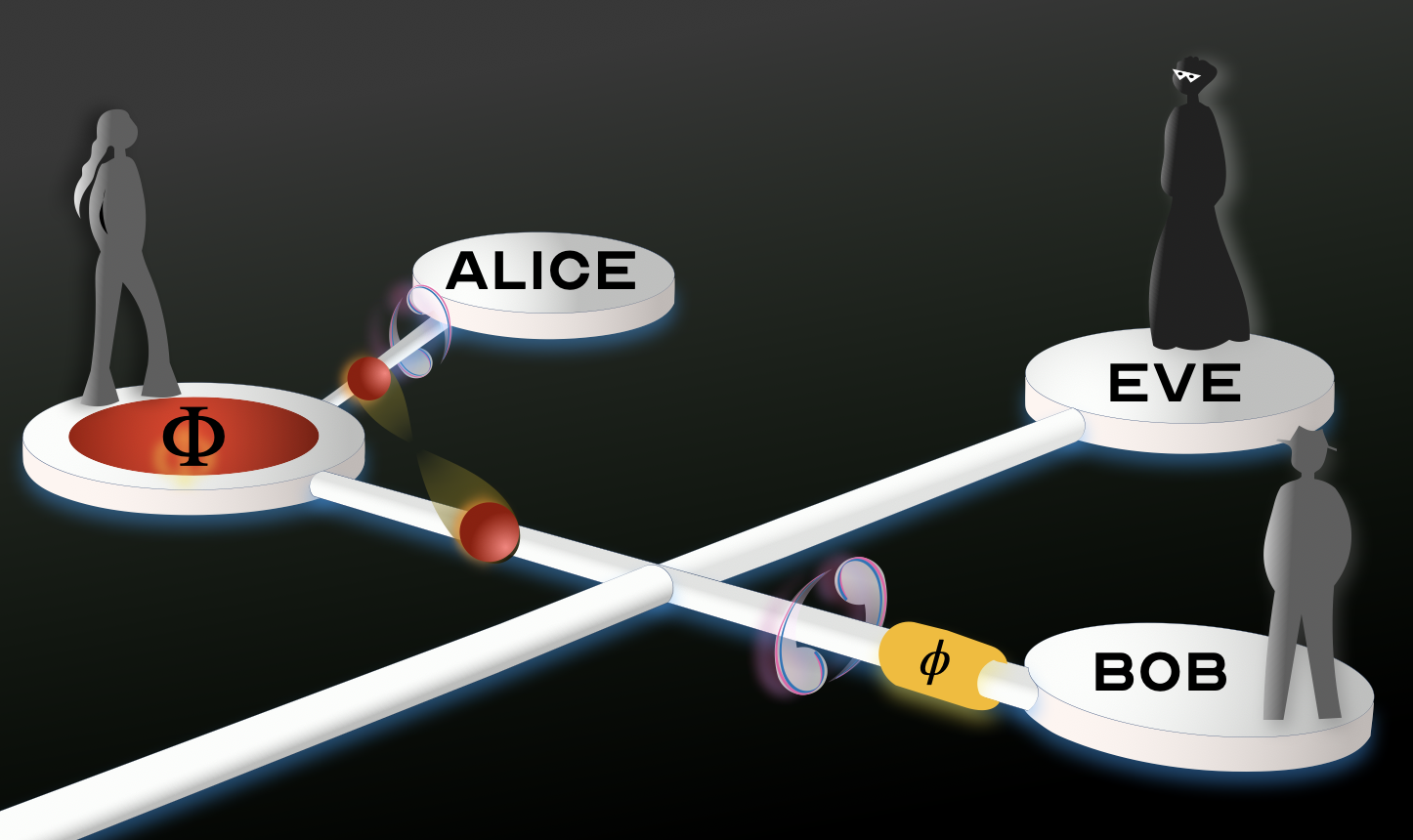}
    \caption{{Concept of our certifiable SPQR setting. Alice is interested in a target parameter $\phi$ only Bob has local access to. Therefore, she distributes part of an entangled state $\Phi$ to Bob over an untrusted quantum channel. At the same time, an adversarial user, Eve, manipulates the network in order to gather as much information as possible about $\phi$ and tamper with the probe received by Bob. We show that Alice and Bob can perform a twirling operation in order to protect against arbitrary channel attacks by Eve, but without compromising the probe's metrological value for sensing.}}
    \label{fig:concept}
\end{figure*}
An adversary, Eve, is interested in retrieving precisely the value of the parameter $\phi$ and disrupting the estimation: she thus controls the quantum network in order to steal the quantum state transmitted through it and tamper with the state received by Bob. As a result of Eve's intervention, without loss of generality, the three share a pure quantum state after the transmission. Alice possesses as tools for her estimation the access to her share of the entangled state and the classical measurement data announced by Bob on a public classical channel. Clearly, in this setting--akin to reverse reconciliation in quantum key distribution (QKD)~\cite{Pirandola2020}--while Eve's and Alice's shares of the state are parameter-independent after Bob's encoding, a public announcement of the measurement outcomes effectively updates their local states with a non-trivial parameter dependence\footnote{Our methods work equally well in the case--akin to direct reconciliation in QKD--where Alice measures and announces outcomes, while Bob uses them to estimate the parameter locally. In this latter case, the estimation is trivially private, but Eve can still affect its precision via tampering.}. In quantum sensing, a central aim is determining the mean square error (MSE) in the estimation of a parameter $\phi$ encoded into a quantum density matrix $\rho(\phi)$. In general terms, this can be done via a positive-operator-valued measurement (POVM) $\{M_{m}\}$, giving the outcome $m$ with probability $p(m|\phi)=\tr{M_{m}\rho(\phi)}$, followed by classical post-processing to produce as estimator $\hat\phi$ with its own distribution $p(\hat \phi)$:  this results in a $\mse{}(\phi) = \int d\hat\phi~ p(\hat \phi) (\hat\phi-\phi)^2$~\cite{helstromBOOK}. If the estimator is unbiased, as it is often assumed, its mean value coincides with the true value, i.e., $\int d\hat\phi~ p(\hat \phi) \hat\phi=\phi$, and the $\mse{}$ corresponds to the ordinary variance.

Differently from the standard setting, in our adversarial case, we have a dual interest. On the one hand, Alice needs to evaluate the MSE for her specific choice of measurement strategy; this can be done via the classical Cram\'er-Rao bound~\cite{helstromBOOK}, which bounds the MSE of all estimators $\hat\phi$ for a fixed measurement:
\begin{equation}\label{eq:FI-general}
     \mse{A}(\phi;\{M_{m}\}) \geq\frac{1}{F(p_A)},
\end{equation}
where $F(p) = \sum_m  p(m|\phi) (\partial_\phi\log p(m |\phi))^2$ is the classical Fisher information (CFI) and $p_A(m|\phi)$ is the probability of Alice's measurement on her state $\rho_A(\phi)$ conditional on the value $\phi$ for the parameter.  On the other hand, we also need to lower-bound the MSE attainable by Eve, irrespectively of her measurement strategy; this can be done via the quantum Cram\'er-Rao bound~\cite{helstromBOOK}, which minimizes the CFI over all POVM's:
\begin{equation}
    \mse{E}(\phi)\geq\frac{1}{Q(\rho_E)},
\end{equation}
where $\rho_E(\phi)$ is Eve's state, $Q(\rho) := \tr{\rho(\phi) L(\phi)^2}$ is the quantum Fisher information (QFI),  and $L(\phi)$ is the symmetric-logarithmic-derivative (SLD) operator, defined implicitly as the solution to the Sylvester equation $\{\rho(\phi),L(\phi)\}=2\partial_\phi\rho(\phi)$. The two problems are deeply correlated in that Alice and Eve share a bipartite state after Bob's parameter encoding, measurement and public announcement. 

The tension between cryptography and metrology arises precisely from this duality of objectives:--\emph{security}--the legitimate users want to randomize Eve's local state, in order to minimize her QFI and ensure she has a large estimation error;--\emph{integrity}--simultaneously, they want to ensure that Alice's state retains a well-characterized parameter-dependence, ideally maximizing her CFI. Clearly, both system randomization and characterization are required--to a certain extent--in order to attain good perfomance under both aspects. Yet, the extent to which these are compatible remains largely unclear to date~\cite{Kianvash2025a,Moore2023,Bizzarri2025a,Moore2025,Ho2024,Huang2019a,Yin2019,Shettell2022a}.  

In our proposal, we solve this issue via bilateral Pauli twirling, a central randomization technique used in QKD in order to reduce channel assumptions~\cite{Shor2000,Kraus2005}. If Alice and Bob share qubit states, via public communication they can pick a Pauli matrix at random and apply it to both their local states, forgetting their choice\footnote{In practice, they need not carry out Pauli twirling on their quantum states during the protocol, but rather Alice can implement it in post-processing via Pauli-frame update, i.e., a simple rotation of the measured Pauli correlators, as detailed below.}. The effective average state resulting from this procedure can be equivalently written in Bell-diagonal or Pauli-diagonal form:
\begin{equation}\label{eq:twirled_State_AB}
    \sigma_{AB} = \sum_{\beta=0}^3 p_\beta \dyad{\beta} \equiv \frac14\sum_{\beta=0}^3 t_\beta P_\beta^A\otimes P_\beta^B,
\end{equation}
where $\{p_\beta\}_{\beta=0}^3$ is a probability distribution, $t_0=1$ and $t_{\beta>0}\in[-1,1]$, while $\ket{\beta}_{AB}=P_\beta^B\ket{\Phi_+}_{AB}$ are the Bell states\footnote{Up to a global phase $\ii$ for $\beta=2$.}, having defined $\ket{\Phi_+} = \tfrac{1}{\sqrt{2}}(\ket{00}+\ket{11})$ and the Pauli matrices $P_\beta=I,X,Y, Z$ for $\beta=0,1,2,3$. Thus, Alice and Bob can maintain a certain degree of quantum correlations, determined by the states' coefficients, while forcing Eve's correlations to have a well-defined form. Indeed, assuming without loss of generality that Eve holds the purifying system, we can write the full pure state of the three parties as
\begin{equation}\label{eq:tripartite}
    \ket{\psi}_{EAB} = \sum_{\beta=0}^3 \sqrt{p_\beta} ~\ket{e_\beta}_E \otimes \ket{\beta}_{AB},
\end{equation}
where $\{\ket{e_\beta}\}_{\beta=0}^3$ is an orthonormal basis of Eve's Hilbert space. 

However, while twirling provides Alice and Bob with the means to quantify Eve's effect on the channel and still maintain some quantum correlations, in general it does not commute with the parameter-encoding process: {\it per se} it does not guarantee that high metrological precision can still be attained.
Surprisingly, we will show that this is the case for families of single-parameter unitary encodings $U_\phi=e^{-\ii\phi G}$, taking as a prototypical example the estimation of a phase with generator $G=\tfrac{Z}{2}$. Clearly, while the optimal probe would be a Bell state, the twirled state \eqref{eq:twirled_State_AB} in general retains phase-sensitivity, as the phase-encoding performs a rotation of the $X, Y$ operators.  

In each run of our remote phase-estimation protocol, Alice initially sends half of a Bell state $\ket{\Phi_+}$ through the insecure channel to Bob, who decides at random whether to encode (\est round) or not to encode (\che round) the phase on his local state; they then measure in random Pauli bases $P_\beta P_{\beta'}$ for $\beta,\beta'=1,2,3$. At the end of the runs, Alice and Bob disclose their measurement bases and results for \che rounds, and estimate the diagonal Pauli correlators--akin to the channel parameter-estimation phase in QKD:
\begin{equation}\label{eq:correlators_zero_phase}
     \langle X\otimes X\rangle_0\equiv t_1, \quad \langle Y\otimes Y\rangle_0\equiv t_2, \quad  \langle Z\otimes Z\rangle _0\equiv t_3,
\end{equation}
where the expectation is taken with respect to the phase-independent state $\sigma_{AB}$.
Based on these values, they can calculate exactly Eve's potential QFI in case the \est rounds were disclosed, as detailed below, and determine whether to abort the protocol or proceed. In the latter case, Bob discloses his measurement basis and results for the remaining \est rounds, which Alice can fruitfully employ for the estimation. 
Indeed, after phase-encoding, the value of the diagonal Pauli correlators on the resulting state  $\sigma_{AB}(\phi) = U_\phi \sigma_{AB} U_\phi^\dagger$ can be computed as $\ave{P_\beta\otimes P_\beta}_\phi = \ave{P_\beta \otimes U_\phi^\dagger P_\beta U_\phi}_0$, yielding a residual phase-dependence for Pauli operators non-commuting with the phase-generator:
\begin{equation}\label{eq:correlators-general}
    \langle X\otimes X\rangle_{\phi} = t_1\cos(\phi), \quad
    \langle Y\otimes Y\rangle_{\phi}=  t_2\cos(\phi), \quad  \langle Z\otimes Z\rangle_\phi = t_3.
\end{equation}
For a single $P_\beta P_\beta$ (with $\beta=1,2$) measurement round, after public announcement Alice can reconstruct the global measurement outcome, which behaves as a Bernoulli random variable with  probability $p_\beta(\pm1|\phi) = \tfrac{1\pm t_\beta \cos(\phi)}{2}$. In practice, an estimator whose MSE approaches the performance dictated by \eqref{eq:FI-general} asymptotically in the number of rounds is the maximum-likelihood estimator
\begin{equation}
\label{eq:estimator}
    \hat\phi(t_1,t_2) = \arccos\left(\frac{\langle X\otimes X\rangle_{\phi}}{2 t_1} + \frac{\langle Y\otimes Y\rangle_{\phi}}{2 t_2} \right).
\end{equation}

We provide a demonstration of the protocol at the proof-of-principle level, in order to illustrate the key steps of the certification. The entanglement source is built by employing a photonic quantum gate, see Fig.~\ref{fig:scheme}. Two qubits are written on the polarisation of two photons from the same parametric down-conversion event. The gate itself comprises a partially polarising beam splitter (transmission $T_H\simeq1$ for the horizontal polarisation, $T_V\simeq0.33$ for the vertical polarisation), realising two-photon interference. This effect results in  a $\pi$ phase shift to the component of the quantum state in which both photons are vertically polarised {\color{red}~\cite{Langford2005}}. By choosing appropriate superpositions as the qubit inputs, the entangled state
\begin{equation}
\label{eq:state}
    \ket{\Phi_+}=\frac{1}{\sqrt{2}}\left(\ket{HD}+\ket{VA}\right),
\end{equation}
can be targeted. Here, the notation $\ket{HD}$ ($\ket{VA}$) denotes Alice's photon in the horizontal (vertical) polarisation, and Bob's in the diagonal (antidiagonal) one. Therefore, the encoding of the Pauli matrices follows the following rules: for Alice $X=\ket{H}\bra{H}-\ket{V}\bra{V}$, $Y=\ket{D}\bra{D}-\ket{A}\bra{A}$, while for Bob $X=\ket{D}\bra{D}-\ket{A}\bra{A}$, $Y=\ket{H}\bra{H}-\ket{V}\bra{V}$. For both parties, $Z=\ket{L}\bra{L}-\ket{R}\bra{R}$ in the basis of the Right- and Left-circular polarisations. Imperfect two-photon interference results in a reduction of the coherences in the state \eqref{eq:state}, which we attribute in our analysis to the action of the malevolent Eve.

\begin{figure}
    \centering
    \includegraphics[width=\textwidth]{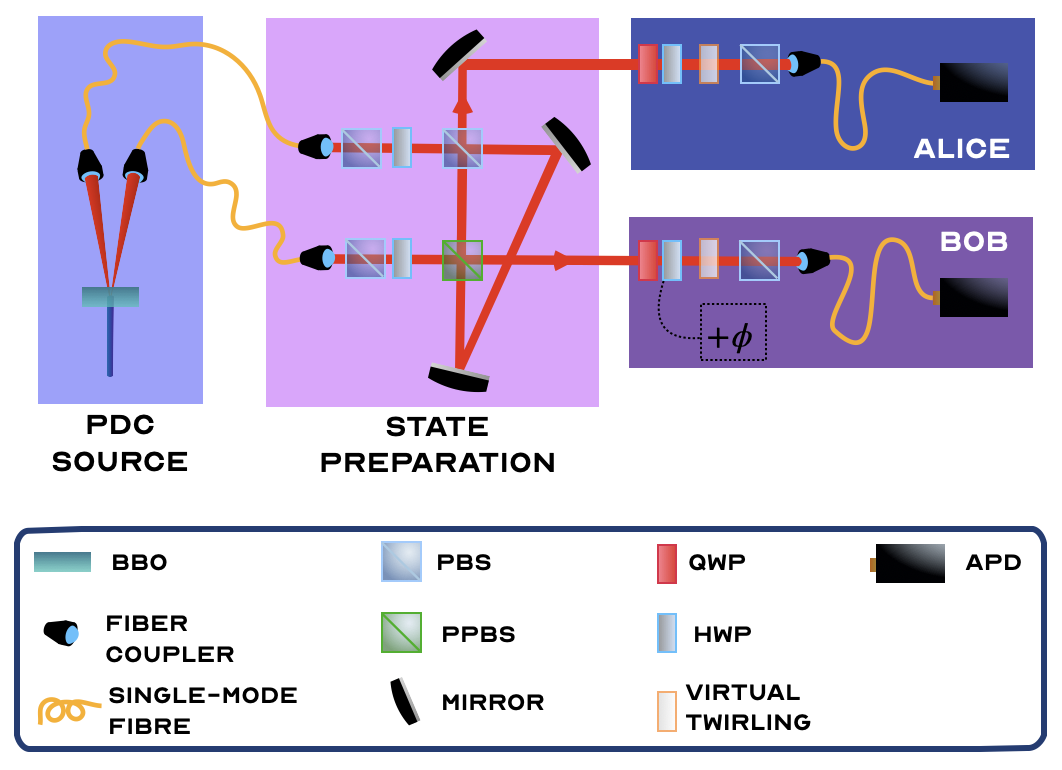}
    \caption{{Photonic platform for SPQR. Photonic qubits are produced by means of a parametric down conversion (PDC) process in a $3\si{\milli\meter}$-thick $\beta$-barium borate (BBO) crystal pumped by a $50 \si{\milli\watt}$-continuous-wave (CW) laser with central wavelength $\lambda_0 = 405 \si{\nano\meter}$. Polarization-encoded qubits are then coupled by means of a photonic gate, producing the state \eqref{eq:state}., starting from the polarisation $\cos (\pi/3)\ket{H}+\sin (\pi/3)\ket{H}$ on both inputs.  The gate comprises a partially polarising beam splitter (PPBS), which allows for polarisation-selective nonclassical interference. This is embedded within a Sagnac interferometer, which reduces the impact of deviations of the actual transmittivities from the ideal ones. Following the generation, the polarization of both Alice's and Bob's photons is measured by means of standard polarimeters, made of a quarter-wave plate (QWP), a half-wave plate (HWP) and a polarising beam splitter (PBS). The target value of the interferometric phase $\phi$ is here encoded in Bob's measurement station. Single-mode fibers are used throughout the apparatus in order to achieve good indistinguishability of the spatial modes. Avalanche photodiodes (APD) finally register the arrival of photons, with only coincidence events being considered - postselection is intrinsic to the functioning of the gate. Pauli twirling is carried out virtually during post-processing of the measurement outcomes.}}
    \label{fig:scheme}
\end{figure}

As customary in the analysis for QKD, the twirling operations are not physically implemented, but rather executed by opportune relabelling. Specifically, given our choice of observables and the correlations in the local Pauli transformations, our observables actually remain unaffected by the twirling. A \che round thus delivers channel parameter estimations for $t_1$, $t_2$ and $t_3$; in order to assess their variability, we adopt $R_{\rm che}=200$ repetitions of the reconstruction of $\langle X\otimes X\rangle$, $\langle Y\otimes Y\rangle$, and $\langle Z\otimes Z\rangle$ by standard polarimetry. We have obtained as average values and standard deviations of the samples: $\bar t_1\pm\Delta t_1=0.911\pm0.043$,  $\bar t_2\pm\Delta t_2=-0.926\pm0.037$, and $\bar t_3\pm\Delta t_3=0.913\pm0.039$.

The parameter of interest is a rotation $\theta$ of Bob's reference, which is equivalent to changing the position of the polarimeter's plates accordingly. In formal terms, this behaves as a phase shift by $\phi=4\theta$~\cite{Langford2005}. 
An \est round thus consists in the measurements of $c_1=\langle X\otimes X\rangle_\phi$ and $c_2=\langle Y\otimes Y\rangle_\phi$ taken with the different orientation, based on a number of $N_{\rm est}$ events. As above, in order to ascertain the experimental uncertainty, we have taken the average and standard deviation of such correlators from $R_{\rm est}=200$ repetitions. These are inserted in the estimator \eqref{eq:estimator} in order to extract a value for the phase $\phi$. We observe that, while in our experiment we performed \che and \est rounds sequentially, for on-field deployments these should be alternated randomly, so as not to leak information to Eve.  The retrieved values of $\phi$ are reported in Fig.~\ref{fig:phases} as a function of the rotation $\theta$: the data are well in agreement with the theoretical prediction within their uncertainties - these will be discussed in detail in the following section - thus we can consider our estimator as unbiased. This demonstrates that the obfuscation provided by twirling can be compatible with the level of system characterisation needed for parameter retrieval.  

\begin{figure}
    \centering
    \includegraphics[width=\textwidth]{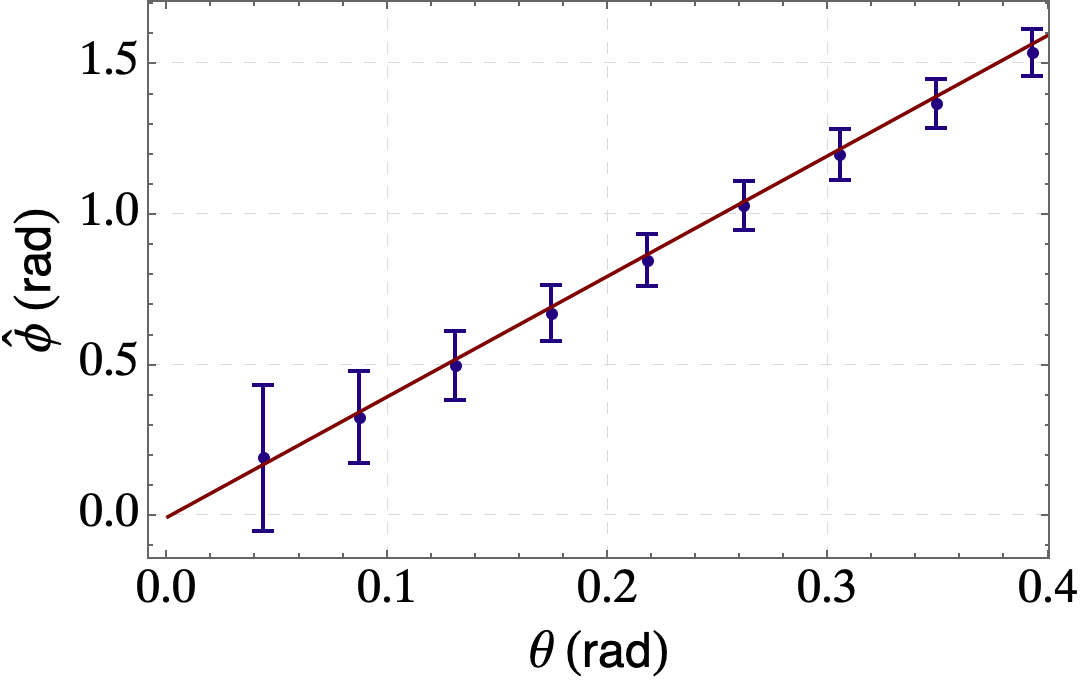}
    \caption{ Estimated phases as a function of the rotation angle $\theta$. Blue points and vertical bars correspond to experimental data with the associated experimental uncertainty. Theoretical prediction is here represented by the red continuous line.}
    \label{fig:phases}
\end{figure}

\subsection{Certification of security and integrity
in remote phase-estimation}

We have thus shown that, thanks to bilateral Pauli twirling, the effect of Eve's tampering is constrained to mixing in the Bell basis. Consequently, it can be estimated precisely in the \che rounds, providing a certification of its impact on the estimation variance, in contrast with previous approaches~\cite{Kianvash2025a,Moore2023,Bizzarri2025a,Moore2025,Ho2024,Huang2019a,Yin2019,Shettell2022a}.   

Indeed, Alice's CFI on $\phi$ can be computed via error-propagation as \begin{equation}\label{eq:FI-single}
    F(p_\beta) = \frac{t_\beta^2\sin^2(\phi)}{1 - t_\beta^2\cos^2(\phi)},
\end{equation}
based on the expression for her conditional detection probabilities $p_\beta(\pm1|\phi)$. Therefore, for a number $n_X$ of $XX$ rounds and $n_Y$ of $YY$ rounds, the CFI of Alice is $F_A(\phi; t_1,t_2) = n_X\,F(p_1) + n_Y\,F(p_2)$, with $n_X+n_Y=N_{\rm est}$.

We now move on to describe synthetically how Alice and Bob can bound the best-case estimation performance of Eve based on the correlators \eqref{eq:correlators_zero_phase} measured during the \che phase, deferring its detailed calculation to the Appendix; 
for simplicity, we assume that Eve performs independent and identically distributed (iid) attacks on each run\footnote{This assumption can be dropped via a quantum de-Finetti argument by discarding a suitable number of rounds, as already discussed in~\cite{Bizzarri2025a}.}. Starting from the tripartite state \eqref{eq:tripartite}, and defining Bob's POVM   $\{\Pi_m\}$, her state conditional on Bob's public announcement is
\begin{equation}\label{eq:conditional_eve_state}
  \rho_{E\mid m}(\phi)
  = \frac{1}{p(m)} \sum_{\beta,\beta'} \sqrt{p_\beta p_{\beta'}}
  \,\langle \beta' | (I\otimes \Xi_m(\phi)) | \beta\rangle\,
  \dyad{e_\beta}{e_{\beta'}},
\end{equation}
where  $\Xi_m(\phi):=U_\phi^\dagger \Pi_m U_\phi$ and the outcome probability is clearly $\phi$-independent:
\begin{equation}
    p(m)=\sum_\beta p_\beta\,\langle \beta | (I\otimes \Xi_m(\phi)) | \beta \rangle
    =\frac12\Tr \Xi_m(\phi)=\frac12{\rm rank}(\Pi_m).
\end{equation}
However,  the post-selected state $\rho_{E\mid m}(\phi)$ generally does depend on $\phi$ non-trivially.
Therefore, Eve’s QFI after the announcement,
$\sum_m p(m)\,Q\big(\rho_{E\mid m}\big)$,
is in general strictly positive.

When Bob measures $X$ with projectors $\Pi_\pm=\tfrac12(I\pm X)$, we have $\Xi_\pm(\phi)=\tfrac12\bigl(I\pm X\cos\phi \pm Y\sin\phi\bigr)$. 
The matrix elements of \eqref{eq:conditional_eve_state} can be computed tediously but straightforwardly by noting that $\langle \beta' | (I\otimes \Xi) | \beta \rangle =\tfrac12\,\tr{\Xi\,P_\beta\, P_{\beta'}}$, yielding
\begin{equation}\label{eq:E-unnormalized-matrix} 
    \rho_{E\mid m}(\phi) =  
    \begin{pmatrix} 
    p_0 & m \sqrt{p_0 p_1}\,\cos\phi & m \sqrt{p_0 p_2}\,\sin\phi & 0 \\ 
    m \sqrt{p_0 p_1}\,\cos\phi & p_1 & 0 & -\ii m\sqrt{p_1 p_3}\,\sin\phi \\ 
    m \sqrt{p_0 p_2}\,\sin\phi & 0 & p_2 & \ii m\sqrt{p_2 p_3}\,\cos\phi \\ 
    0 & \ii m \sqrt{p_1 p_3}\,\sin\phi & -\ii m \sqrt{p_2 p_3}\,\cos\phi & p_3
    \end{pmatrix}\!
\end{equation}
in the $\{\ket{e_\beta}\}$ basis, having noticed that $p(m)=1/2$.  

By direct evaluation, it can be checked that the matrix has rank two and can be written as $\rho_{E|m}(\phi)=\dyad{\varphi_0(\phi)}+\dyad{\varphi_1(\phi)}$, with 
\begin{align}
\ket{\varphi_{0|m}(\phi)}
&:= \sqrt{p_0}\ket{e_0}\ +m \Big(\sqrt{p_1}\cos\phi\,\ket{e_1}+ \sqrt{p_2}\sin\phi\,\ket{e_2}\Big),\nonumber\\
\ket{\varphi_{1|m}(\phi)}
&:= -\ii\sqrt{p_3}\ket{e_3} +m\Big(-\sqrt{p_1}\sin\phi\,\ket{e_1}+ \sqrt{p_2}\cos\phi\,\ket{e_2}\Big).\label{eq:vectors_dec_rho}
\end{align}
non-orthogonal, unnormalized pure states, which simplifies numerical evaluation of its QFI via Gram matrices~\cite{Rosati2021}.

In order to write a simple analytical expression for the QFI, we note further that the vectors \eqref{eq:vectors_dec_rho} can be obtained via the action of a covariant operation in the subspace ${\rm span}\{\ket{e_1},\ket{e_2}\}$, followed by a phase-independent shearing. Therefore, defining the subspace $y$-rotation
\begin{equation}
    R(\phi) = e^{-\ii \phi Y_{12}} = \cos\phi\, I_{12} - \ii\sin\phi\, Y_{12},
\end{equation}
and the shearing operators
\begin{equation}
   B_1 =  \sqrt{p_0} \dyad{e_0} + \sqrt{p_3} \dyad{e_3} , B_2 = \sqrt{p_1} \dyad{e_1} + \sqrt{p_2} \dyad{e_2},
\end{equation}
we can write $\rho_{E|m}(\phi) = V(\phi) V(\phi)^\dagger$ with $V(\phi) = B_1 \oplus B_2 R(\phi)$. As shown in the Appendix, for small phase-shifts this system is analogous to a pure probe state under standard unitary evolution. Hence its QFI is $\phi$-invariant and equals
\begin{align}
Q_E(\rho_{E|m})&=
1-t_3^2,\label{eq:final_eve_qfi}
\end{align}
which is also valid for $Y$ measurements.

Surprisingly, we uncover that Eve's QFI is uniquely determined by the $ZZ$ correlator, which can be estimated during the \che rounds and used to certify security. This can be understood from two different perspectives about the twirled state $\sigma_{AB}$ \eqref{eq:twirled_State_AB}: 
in the Pauli basis, a large weight of the phase-insensitive ZZ component implies little information about the parameter; in the Bell basis, maximizing the presence of Bell states related via a phase-sensitive $X$ operation favours Eve's estimation. 

Yet, in the experiment, we are faced with the additional difficulty that the estimator \eqref{eq:estimator} depends explicitly on the channel parameters $t_1$, $t_2$; this is a special case of multiparameter estimation in which the parameters are retrieved sequentially rather than jointly~\cite{Mukhopadhyay2025}. The operation interwines the \che and \est steps, that both contribute to the uncertainty on the retrieved phase values. Therefore, the estimation variance is obtained by error propagation as
\begin{equation}
\begin{aligned}
    \Delta^2\phi &= \left(\partial_{t^{\rm che}_1}\hat \phi\right)^2\Delta^2 t^{\rm che}_1+\left(\partial_{t^{\rm che}_2}\hat \phi\right)^2\Delta^2 t^{\rm che}_2+\left(\partial_{t^{\rm est}_1}\hat \phi\right)^2\Delta^2t^{\rm est}_1+\left(\partial_{t^{\rm est}_2}\hat \phi\right)^2\Delta^2t^{\rm est}_2\\
    &=\Delta^2\phi_{\rm est}+\Delta^2\phi_{\rm che},
    \end{aligned}
    \label{eq:variances}
\end{equation}
where we have distinguished the contributions from the two classes of rounds. Notice that, since Eve has complete control over the channel, she does not need to learn the $t_\beta$'s from the data.  Eq.~\eqref{eq:variances} provides the uncertainties reported as a function of $\theta$ in Fig.~\ref{fig:variances}, which have also been used for the error bars in Fig.~\ref{fig:phases}.  The squares correspond to the total variance $N_{\rm est} \Delta^2\phi$, rescaled to the number of resources, and follow a descending trend with $\theta$. The comparison with Eve's attainable precision $1/Q_E(\rho_{E|m})$ shows a region in which Alice's precision is superior to Eve's: its extent is determined by the confidence level Alice sets on Eve's information, with the solid black line indicating the estimated QFI, and the shadings accounting for increasingly stringent requirements. These experimental values follow well the prediction of a multiparameter treatment (dashed red line, see Appendix), while, isolating the contributions from  $N_{\rm est}\Delta^2\phi_{\rm est}$, these are close to the single-parameter CRB defined by means of Alice's CFI in \eqref{eq:FI-single}.

\begin{figure}
    \centering
    \includegraphics[width=\textwidth]{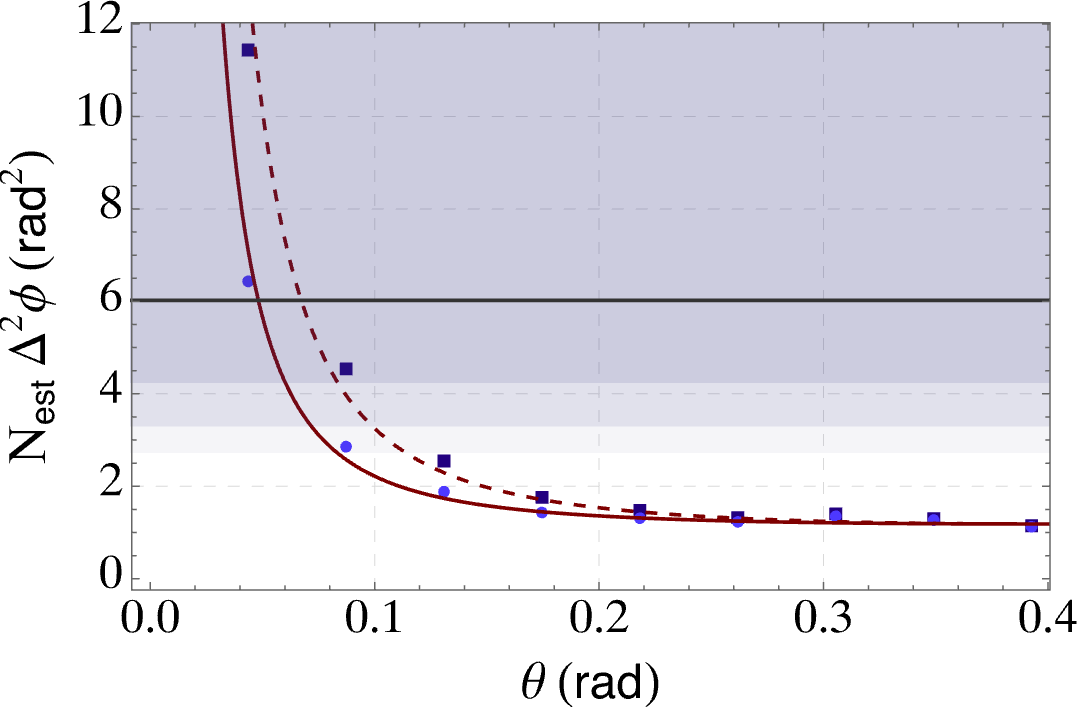}
    \caption{ Estimation performances in presence of a malicious adversary. Squares correspond to total estimation uncertainties, retrieved by including in $\Delta^2\phi$ uncertainties in the \che and \est rounds by error propagation. Circles represent estimation uncertainties $\Delta^2\phi_{\rm est}$ retrieved by considering the \est rounds only for error propagation.  On average, $ N_{\rm res}=193$ resources have been used in a single round. Accordingly, the red continuous line is associated with Alice's CRB at fixed \che round correlators, whereas the red dashed line depicts Alice's CRB  from a multiparameter treatment. The black, continuous line represents an estimate of Eve's QFI, according to \eqref{eq:final_eve_qfi}. The gray shadings coincide with progressively larger deviations from $\bar{t}_3$, by 1 to 3 standard deviations. }
    \label{fig:variances}
\end{figure}

\section{Discussion}
We have established and demonstrated a SPQR protocol that, for the first time to our knowledge, enables the simultaneous certification of precision--for the legitimate users--and security--against the eavesdropper--of the remotely estimated phase. In the spirit of network applications, where quantum technological capabilities may vary greatly among different parties, our protocol has minimal technical requirements. Strikingly, we showed that the use of cryptographic randomization methods such as Pauli twirling can constrain the eavesdropper's effective attacks without reducing the metrological performance of the effective probe state. Therefore, while twirling is used in cryptography in order to reduce the assumptions on the system, in this context it turns out to be useful in sensing for conceding a minimal amount of assumptions, allowing the two worlds to meet successfully halfway. 

While our protocol opens a concrete path forward to implement private remote sensing in real quantum networks, it remains an open problem to determine the ultimate joint precision-privacy limits allowed by quantum mechanics. 
At the same time, there remain technological challenges ahead in order to deploy large-scale secure sensing networks, whose practical impact is yet to be understood, but can now be taken up with cautious optimism.

\bmhead{Acknowledgements} We thank Ilaria Gianani, Damian Markham, Sean Moore and Jacob Dunnigham for valuable discussion. This work has been realised with the support of the PRIN project PRIN22-RISQUE-2022T25TR3 of the Italian Ministry of University and by the the NQSTI (Bando a Cascata CUP: B53C22004170006). M.R. acknowledges support from MUR (project FIS-2023-03472, “SuNRISE” CUP F53C24001760001). M.B. and G.B. acknowledge support by Rome Technopole Innovation Ecosystem (PNRR grant M4-C2-Inv) and MUR Dipartimento di Eccellenza 2023-2027. \bmhead{Data availability} The data supporting this work are available at the public repository 10.5281/zenodo.20597482 




\section{Appendix}

\subsection{Calculation of the QFI for Eve's sheared covariant family}
First, observe that, when Bob measures $Y$ with projectors $\Pi_\pm^Y=\tfrac12(I\pm Y)$ we have
$\Xi_\pm^Y(\phi)=\tfrac12\big(I\pm Y\cos \phi \mp X\sin \phi\big)$. Thus, the QFI for a $Y$ measurement is the same as for a $X$ measurement, and in the following we focus on the $X$-measurement.

Eve's matrix coefficients can be calculated using $\Xi_\pm(\phi)=\tfrac12(I\pm X\cos\phi \pm Y\sin\phi)$ and the Pauli matrices' commutation relations and trace properties, yielding
\begin{equation}\label{eq:basic-traces} 
    \tfrac12\operatorname{Tr}[\Xi_\pm(\phi)\,I]=\tfrac12,\quad \tfrac12\operatorname{Tr}[\Xi_\pm(\phi)\,X]=\pm\tfrac12\cos\phi,\quad \tfrac12\operatorname{Tr}[\Xi_\pm(\phi)\,Y]=\pm\tfrac12\sin\phi, 
\end{equation} while the traces with $Z$ vanish. Taking into account that $P_\beta P_{\beta'}$ is proportional to one of the Pauli matrices, all matrix elements of \eqref{eq:E-unnormalized-matrix} are proportional to one of the traces \eqref{eq:basic-traces}.


We now prove that Eve's conditional state after Bob's measurement and public announcement has QFI that is phase-invariant, despite it being a non-unitary encoding of the parameter itself. Our proof relies on decomposing the effective parameter-encoding process as a two-dimensional rotation followed by a non-unitary coordinate rescaling, or shearing, in a suitable two-dimensional subspace. It then turns out that the infinitesimal change in fidelity due to these processes (or equivalently the associated Bures metric) is analogous to that of a pure probe state under a suitable unitary encoding of the parameter, which is notoriously parameter-invariant.

Since our observations apply to a wider family of parameter encodings, that we call sheared covariant families, we state our main result in its most general form, and prove it at the end of this section:

\begin{theorem}[Invariant QFI of rank-two sheared covariant encodings]
\label{thm:rank_two_rotation_rescaling}
Let $\mathcal H=\bigoplus_{\ell=1}^N \mathcal H_\ell$ be an orthogonal direct sum of Hilbert spaces. For each \(\ell\), let $B_\ell:\mathbb C^2\to \mathcal H_\ell$
be a linear map such that the \(2\times2\) Gram matrix $G_\ell:=B_\ell^\dagger B_\ell$ is real symmetric. Let
\begin{equation}
R_\theta:=
\begin{pmatrix}
\cos\theta&-\sin\theta\\
\sin\theta&\cos\theta
\end{pmatrix}
\end{equation}
and let \(\omega_\ell\in\mathbb R\). Define the amplitudes matrix
\begin{equation}
V_\phi
:=
\bigoplus_{\ell=1}^N B_\ell R_{\omega_\ell\phi},
\end{equation}
and the normalized state
\begin{equation}
\rho_\phi
:=
\frac{V_\phi V_\phi^\dagger}{T},
\qquad
T:=\Tr[V_\phi^\dagger V_\phi].
\end{equation}
If \(V_\phi\) has column rank two for all \(\phi\) in the parameter
region considered, then the SLD quantum Fisher information of \(\rho_\phi\)
is independent of \(\phi\) and is given by
\begin{equation}\label{eq:qfi_general_sheared_covariant_main}
Q(\rho_\phi)
=
4\sum_{\ell,j<\ell}q_\ell q_j (\omega_\ell-\omega_j)^2
\end{equation}
where $q_\ell:=\frac{\Tr G_\ell}{T}$.

\end{theorem}
In our specific case, we have that the Theorem applies with $N=2$, $\omega_1=0$, $\omega_2=1$ and the choice of linear operators
\begin{equation}
B_{1}
=
\begin{pmatrix}
\sqrt{p_0}&0\\
0&-\ii\sqrt{p_3}
\end{pmatrix}, \quad  B_{2}
= m
\begin{pmatrix}
\sqrt{p_1}&0\\
0&\sqrt{p_2}
\end{pmatrix},
\end{equation}
acting respectively in $\mathcal H_1 = \mathrm{span}\{\ket{e_0},\ket{e_3}\}$ and $\mathcal H_2 =\mathrm{span}\{\ket{e_1},\ket{e_2}\}$. It follows $q_1 = p_0+p_1$, $q_2=p_1+p_2$ and $T=1$ and therefore the QFI expression \eqref{eq:qfi_general_sheared_covariant} yields
\begin{equation}\label{eq:qfi_bell_basis}
    Q(\rho_\phi) = 4 q_1 q_2 = 4(p_0+p_1)(p_1+p_2).
\end{equation}
Finally, the relation between the Bell-basis and Pauli-basis coefficients can be determined by computing the matched Pauli correlators in the Bell basis
\begin{align}
    t_i=4\ave{P_i \otimes P_i}_\sigma = \sum_{\beta} p_\beta \bra{\beta}P_i \otimes P_i\ket{\beta} = \sum_{\beta} p_\beta \bra{0}P_i \otimes P_\beta P_i P_\beta\ket{0}.
    \end{align}
Recalling that $P_\beta P_i P_\beta = -P_i$ for $i>0$ and $\beta\neq i$, due to anti-commutativity of the Pauli matrices, and calculating explicitly the correlators on the first Bell state, i.e, $\bra0 P_i P_i\ket0 = 1$ for $i\neq 2$ while  $\bra0 P_2 P_2\ket0 = -1$, since $Y=\ii XZ$, we obtain the relations
\begin{align}
    t_0 &= 4(p_0 + p_1 + p_2 + p_3),\quad t_1 = 4(p_0 + p_1 - p_2 - p_3),\\
    t_2 &= 4(-p_0 + p_1 - p_2 + p_3),\quad   t_3 = 4(p_0 - p_1 - p_2 + p_3).
\end{align}
Their inversion yields
\begin{align}
    p_0 &= \frac14(1 + t_1 - t_2 + t_3),\quad p_1 = \frac14(1 + t_1 + t_2 - t_3),\\
    p_2 &= \frac14(1 - t_1 - t_2 - t_3),\quad   p_3 = \frac14(1 - t_1 + t_2 + t_3),
\end{align}
ans substituting into \eqref{eq:qfi_bell_basis}.

\begin{proof}(Theorem 1)
First observe that
\begin{equation}
T=\Tr[V_\phi^\dagger V_\phi]
=
\sum_{\ell=1}^N \Tr[R_{\omega_\ell\phi}^T
G_\ell R_{\omega_\ell\phi}]
=
\sum_{\ell=1}^N \Tr G_\ell
=
\sum_{\ell=1}^N t_\ell,
\end{equation}
so the normalization is independent of \(\phi\).

We compute the fidelity between two nearby states. By Uhlmann's theorem,
for two states written as \(VV^\dagger\) and \(WW^\dagger\),
\begin{equation}
\mathcal F(VV^\dagger,WW^\dagger)
=
\|V^\dagger W\|_1^2,
\end{equation}
where \(F\) denotes the squared Uhlmann fidelity. Hence
\begin{equation}
F(\rho_\phi,\rho_{\phi+\epsilon})
=
\frac{1}{T^2}
\left\|V_\phi^\dagger V_{\phi+\epsilon}\right\|_1^2 .
\end{equation}
Setting $M_\phi(\epsilon):=V_\phi^\dagger V_{\phi+\epsilon}$ and using the orthogonal block structure of \(V_\phi\), we obtain
\begin{equation}
M_\phi(\epsilon)
=
\sum_{\ell=1}^N
R_{\omega_\ell\phi}^T
G_\ell
R_{\omega_\ell(\phi+\epsilon)} .
\end{equation}

Now decompose each real symmetric Gram matrix as $
G_\ell=\frac{t_\ell}{2}I+S_\ell$,
where \(S_\ell\) is real symmetric and traceless. Then
\begin{equation}
M_\phi(\epsilon)
=
\frac12\sum_{\ell=1}^N t_\ell R_{\omega_\ell\epsilon}
+
\sum_{\ell=1}^N
R_{\omega_\ell\phi}^T
S_\ell
R_{\omega_\ell(\phi+\epsilon)} .
\end{equation}
The first term is the rotation-dilation part of the infinitesimal change. The second term
is a symmetric traceless shear part. Indeed, if \(S\) is real symmetric and
traceless, then $R_\alpha^T S R_\beta$ is again real symmetric and traceless, since one has
\begin{equation}
R_\alpha^T S R_\beta
=
(R_\alpha^T S R_\alpha)R_{\beta-\alpha},
\end{equation}
and multiplying a real symmetric traceless \(2\times2\) matrix by a real
rotation again gives a real symmetric traceless \(2\times2\) matrix.

Therefore \(M_\phi(\epsilon)\) has the effective qubit form
\begin{equation}
M_\phi(\epsilon)=x(\epsilon)I+y(\epsilon)J+S_{\phi,\epsilon},
\end{equation}
where $J=XZ=-iY$, \(S_{\phi,\epsilon}\) is real symmetric and traceless, and
\begin{equation}
x(\epsilon)
=
\frac12\sum_{\ell=1}^N t_\ell\cos(\omega_\ell\epsilon),
\qquad
y(\epsilon)
=
\frac12\sum_{\ell=1}^N t_\ell\sin(\omega_\ell\epsilon).
\end{equation}

Writing
\begin{equation}
S_{\phi,\epsilon}
=
\begin{pmatrix}
u&v\\
v&-u
\end{pmatrix},
\end{equation}
we have that
\begin{equation}
M_\phi(\epsilon)
=
\begin{pmatrix}
x+u&v-y\\
v+y&x-u
\end{pmatrix}.
\end{equation}
For a real \(2\times2\) matrix of this form,
\begin{equation}
\Tr[M^T M]
=
2(x^2+y^2+u^2+v^2),
\end{equation}
and
\begin{equation}
\det M
=
x^2+y^2-u^2-v^2.
\end{equation}
Moreover, for any real \(2\times2\) matrix, writing the trace-norm as the sum of its singular values, it can be easily shown that
\begin{equation}
\|M\|_1^2
=
\Tr[M^T M]+2|\det M|,
\end{equation}
where the first term on the right-hand side equals the sum of squares of the singular values, and the second one their double-product.
Since \(V_\phi\) has column rank two,
\begin{equation}
M_\phi(0)=V_\phi^\dagger V_\phi
\end{equation}
is positive definite. Thus \(\det M_\phi(0)>0\), and by continuity
\(\det M_\phi(\epsilon)>0\) for all sufficiently small \(\epsilon\), in which case it holds
\begin{equation}
\|M_\phi(\epsilon)\|_1^2
=
\Tr[M_\phi(\epsilon)^T M_\phi(\epsilon)]
+
2\det M_\phi(\epsilon).
\end{equation}
Substituting the expressions above, the shear variables \(u,v\) cancel:
\begin{equation}
\|M_\phi(\epsilon)\|_1^2
=
4(x(\epsilon)^2+y(\epsilon)^2).
\end{equation}
Therefore
\begin{equation}\label{eq:fidelity_infinitesimal}
 F(\rho_\phi,\rho_{\phi+\epsilon})
=
\left|
\sum_{\ell=1}^N q_\ell e^{i\omega_\ell\epsilon}
\right|^2 .
\end{equation}

Interestingly, we observe that this fidelity is the same as the one obtained using a pure probe state and a suitable unitary parameter encoding:
\begin{equation}
\ket{\chi_\phi}
=
\sum_{\ell=1}^N \sqrt{q_\ell}\,
e^{i\omega_\ell\phi}\ket{\ell}.
\end{equation}
Indeed, $|\langle\chi_\phi|\chi_{\phi+\epsilon}\rangle|^2 \equiv F(\rho_\phi,\rho_{\phi+\epsilon}) $ of \eqref{eq:fidelity_infinitesimal}, which implies that the QFI is parameter-invariant.

Calculating the square and expanding for small $\epsilon$ we get
\begin{align}
    F(\rho_\phi,\rho_{\phi+\epsilon}) &= \left(\sum_{\ell} q_\ell \cos(\omega_\ell \epsilon)\right)^2 + \left(\sum_{\ell} q_\ell \sin(\omega_\ell \epsilon)\right)^2 \\
    &= \sum_{\ell} q_\ell^2 + \sum_{\ell,j\neq \ell} q_\ell q_j\left(\cos(\omega_\ell \epsilon) \cos(\omega_j \epsilon)+\sin(\omega_\ell \epsilon) \sin(\omega_j \epsilon)\right)\\
    &= 1 - 2 \sum_{\ell,j<\ell} q_\ell q_j \left(1-\cos((\omega_\ell-\omega_j)\epsilon)\right)\\
    &= 1- \epsilon^2\sum_{\ell,j<\ell} q_\ell q_j (\omega_\ell-\omega_j)^2 + O(\epsilon^4).
\end{align}
Using the standard relation between squared Uhlmann fidelity and SLD QFI (having also noticed that the state has constant rank, so the non-zero eigenvalues are $\phi$-independent),
\begin{equation}
 F(\rho_\phi,\rho_{\phi+\epsilon})
=
1-\frac{Q(\rho_\phi)}{4}\epsilon^2+o(\epsilon^2),
\end{equation}
we obtain the final QFI expression \eqref{eq:qfi_general_sheared_covariant_main}, which is $\phi$-independent as expected from the pure-state analogy.
\end{proof}

\subsection{Implementing bilateral Pauli twirling via Pauli-frame tracking}\label{app:pauli_frame_tracking}
Let $\cP(\cdot)=\frac14\sum_{\beta=1}^4 (P_\beta\otimes P_\beta)\cdot(P_\beta\otimes P_\beta)$ denote the bilateral Pauli twirling channel. For any POVM element $M$ and state $\rho$, it holds
\begin{equation}
\mathrm{Tr}\big[(P\!\otimes\!P)\rho(P\!\otimes\!P)\,M\big]
= \mathrm{Tr}\big[\rho\,(P\!\otimes\!P)\,M\,(P\!\otimes\!P)\big].
\end{equation}
That is, the outcome probabilities obtained by twirling the state equal those obtained by twirling the measurement. This permits a virtual implementation of the twirl: we keep the hardware unchanged and instead classically relabel the observed outcomes according to how Pauli conjugation transforms the measured observables. Thus, twirling can be carried out locally and offline by Alice without any public communication once she receives Bob's measurement bases and observables. 

Specifically, the non-trivial conjugation rules are given by $P_{\beta'} P_\beta P_{\beta'} = - P_\beta $ for $\beta\neq\beta'>1$.
Hence for matched two-qubit measurements $P_\beta P_\beta$ we have for all $\beta$ that
\begin{equation}
 (P\!\otimes\!P)^{\!\dagger}(P_\beta\!\otimes\!P_\beta)(P\!\otimes\!P) = P_\beta\otimes P_\beta.
\end{equation}
In other words, the only effect of $P$ corresponds operationally to flipping each local outcome bit, which does not change the parity of matched measurements. Therefore, the analysis in the main text follows with correlators as chosen. Non-matched correlators, if needed, could be updated via suitable sign-flips.

\subsection{Multiparameter treatment}

The measurements considered in our treatment are characterised by a distribution of outcomes $p_\beta(\pm1|\phi) = \tfrac{1\pm t_\beta \cos(\phi)}{2}$. Standard treatment of the Fisher information for the joint estimation of $\phi$ and $t_\beta$ involves the matrix
\begin{equation}
    F_{\rm MP}=
    \begin{pmatrix}
        \frac{t_\beta^2\sin^2\phi}{1-t_\beta^2\cos^2\phi}&-\frac{t_\beta^2\sin^2 2\phi}{2(1-t_\beta^2\cos^2\phi)}\\
        -\frac{t_\beta^2\sin^2 2\phi}{2(1-t_\beta^2\cos^2\phi)}&\frac{\cos^2\phi}{1-t_\beta^2\cos^2\phi}
    \end{pmatrix}.
\end{equation}
This quantifies the information made available by each run of the experiment. The total information is thus the sum of $N_{\rm est}$ such identical terms from the \est rounds, and the initial information on $t_\beta$ available from the \che rounds:
\begin{equation}
    F_{\rm tot}=N_{\rm est}F_{\rm MP}+
    \begin{pmatrix}
        0&0\\
        0&\frac{1}{\Delta^2t_\beta}
    \end{pmatrix}.
\end{equation}
This also solves the issue of the singularity of $F_{\rm MP}$~\cite{Mukhopadhyay2025}. The expected uncertainty on $\phi$ is thus bounded as $\Delta^2\phi\geq \left(F_{\rm tot}\right)^{-1}_{1,1}$ . The limit reported as the dashed line in Fig.~\ref{fig:variances} is obtained from the average information made available in the two settings - this is permitted since the channel parameters $t_1$ and $t_2$ are statistically independent.

\bibliography{library.bib}


\begin{thebibliography}{50}
\ifx \bisbn   \undefined \def \bisbn  #1{ISBN #1}\fi
\ifx \binits  \undefined \def \binits#1{#1}\fi
\ifx \bauthor  \undefined \def \bauthor#1{#1}\fi
\ifx \batitle  \undefined \def \batitle#1{#1}\fi
\ifx \bjtitle  \undefined \def \bjtitle#1{#1}\fi
\ifx \bvolume  \undefined \def \bvolume#1{\textbf{#1}}\fi
\ifx \byear  \undefined \def \byear#1{#1}\fi
\ifx \bissue  \undefined \def \bissue#1{#1}\fi
\ifx \bfpage  \undefined \def \bfpage#1{#1}\fi
\ifx \blpage  \undefined \def \blpage #1{#1}\fi
\ifx \burl  \undefined \def \burl#1{\textsf{#1}}\fi
\ifx \doiurl  \undefined \def \doiurl#1{\url{https://doi.org/#1}}\fi
\ifx \betal  \undefined \def \betal{\textit{et al.}}\fi
\ifx \binstitute  \undefined \def \binstitute#1{#1}\fi
\ifx \binstitutionaled  \undefined \def \binstitutionaled#1{#1}\fi
\ifx \bctitle  \undefined \def \bctitle#1{#1}\fi
\ifx \beditor  \undefined \def \beditor#1{#1}\fi
\ifx \bpublisher  \undefined \def \bpublisher#1{#1}\fi
\ifx \bbtitle  \undefined \def \bbtitle#1{#1}\fi
\ifx \bedition  \undefined \def \bedition#1{#1}\fi
\ifx \bseriesno  \undefined \def \bseriesno#1{#1}\fi
\ifx \blocation  \undefined \def \blocation#1{#1}\fi
\ifx \bsertitle  \undefined \def \bsertitle#1{#1}\fi
\ifx \bsnm \undefined \def \bsnm#1{#1}\fi
\ifx \bsuffix \undefined \def \bsuffix#1{#1}\fi
\ifx \bparticle \undefined \def \bparticle#1{#1}\fi
\ifx \barticle \undefined \def \barticle#1{#1}\fi
\bibcommenthead
\ifx \bconfdate \undefined \def \bconfdate #1{#1}\fi
\ifx \botherref \undefined \def \botherref #1{#1}\fi
\ifx \url \undefined \def \url#1{\textsf{#1}}\fi
\ifx \bchapter \undefined \def \bchapter#1{#1}\fi
\ifx \bbook \undefined \def \bbook#1{#1}\fi
\ifx \bcomment \undefined \def \bcomment#1{#1}\fi
\ifx \oauthor \undefined \def \oauthor#1{#1}\fi
\ifx \citeauthoryear \undefined \def \citeauthoryear#1{#1}\fi
\ifx \endbibitem  \undefined \def \endbibitem {}\fi
\ifx \bconflocation  \undefined \def \bconflocation#1{#1}\fi
\ifx \arxivurl  \undefined \def \arxivurl#1{\textsf{#1}}\fi
\csname PreBibitemsHook\endcsname

\bibitem[\protect\citeauthoryear{Kimble}{2008}]{Kimble2008}
\begin{barticle}
\bauthor{\bsnm{Kimble}, \binits{H.J.}}:
\batitle{{The quantum internet}}.
\bjtitle{Nature}
\bvolume{453}(\bissue{7198}),
\bfpage{1023}--\blpage{1030}
(\byear{2008})
\doiurl{10.1038/nature07127}
\end{barticle}
\endbibitem

\bibitem[\protect\citeauthoryear{Zhang et~al.}{2026}]{Zhang2026}
\begin{barticle}
\bauthor{\bsnm{Zhang}, \binits{J.}},
\bauthor{\bsnm{Wang}, \binits{L.}},
\bauthor{\bsnm{Hai}, \binits{Y.-J.}},
\bauthor{\bsnm{Zhang}, \binits{J.}},
\bauthor{\bsnm{Chu}, \binits{J.}},
\bauthor{\bsnm{Jiang}, \binits{J.}},
\bauthor{\bsnm{Huang}, \binits{W.}},
\bauthor{\bsnm{Liang}, \binits{Y.}},
\bauthor{\bsnm{Qiu}, \binits{J.}},
\bauthor{\bsnm{Sun}, \binits{X.}},
\bauthor{\bsnm{Tao}, \binits{Z.}},
\bauthor{\bsnm{Zhang}, \binits{L.}},
\bauthor{\bsnm{Zhou}, \binits{Y.}},
\bauthor{\bsnm{Chen}, \binits{Y.}},
\bauthor{\bsnm{Guo}, \binits{W.}},
\bauthor{\bsnm{Linpeng}, \binits{X.}},
\bauthor{\bsnm{Liu}, \binits{S.}},
\bauthor{\bsnm{Ren}, \binits{W.}},
\bauthor{\bsnm{Zhong}, \binits{Y.}},
\bauthor{\bsnm{Niu}, \binits{J.}},
\bauthor{\bsnm{Yuan}, \binits{H.}},
\bauthor{\bsnm{Yu}, \binits{D.}}:
\batitle{{Distributed multi-parameter quantum metrology with a superconducting
  quantum network}}.
\bjtitle{Nat. Commun.}
\bvolume{17}(\bissue{1}),
\bfpage{1825}
(\byear{2026})
\doiurl{10.1038/s41467-026-68535-9}
\end{barticle}
\endbibitem

\bibitem[\protect\citeauthoryear{Ge and Jacobs}{2025}]{Ge2025}
\begin{barticle}
\bauthor{\bsnm{Ge}, \binits{W.}},
\bauthor{\bsnm{Jacobs}, \binits{K.}}:
\batitle{{Heisenberg-Limited Continuous-Variable Distributed Quantum Metrology
  with Arbitrary Weights}}.
\bjtitle{Phys. Rev. Lett.}
\bvolume{135}(\bissue{10}),
\bfpage{100801}
(\byear{2025})
\doiurl{10.1103/jkjj-3gvb}
\end{barticle}
\endbibitem

\bibitem[\protect\citeauthoryear{Kim et~al.}{2024}]{Kim2024}
\begin{barticle}
\bauthor{\bsnm{Kim}, \binits{D.-H.}},
\bauthor{\bsnm{Hong}, \binits{S.}},
\bauthor{\bsnm{Kim}, \binits{Y.-S.}},
\bauthor{\bsnm{Kim}, \binits{Y.}},
\bauthor{\bsnm{Lee}, \binits{S.-W.}},
\bauthor{\bsnm{Pooser}, \binits{R.C.}},
\bauthor{\bsnm{Oh}, \binits{K.}},
\bauthor{\bsnm{Lee}, \binits{S.-Y.}},
\bauthor{\bsnm{Lee}, \binits{C.}},
\bauthor{\bsnm{Lim}, \binits{H.-T.}}:
\batitle{{Distributed quantum sensing of multiple phases with fewer photons}}.
\bjtitle{Nat. Commun.}
\bvolume{15}(\bissue{1}),
\bfpage{266}
(\byear{2024})
\doiurl{10.1038/s41467-023-44204-z}
\end{barticle}
\endbibitem

\bibitem[\protect\citeauthoryear{Liu et~al.}{2024a}]{Liu2024d}
\begin{barticle}
\bauthor{\bsnm{Liu}, \binits{S.}},
\bauthor{\bsnm{Tian}, \binits{Y.}},
\bauthor{\bsnm{Zhang}, \binits{Y.}},
\bauthor{\bsnm{Lu}, \binits{Z.}},
\bauthor{\bsnm{Wang}, \binits{X.}},
\bauthor{\bsnm{Li}, \binits{Y.}}:
\batitle{{Integrated quantum communication network and vibration sensing in
  optical fibers}}.
\bjtitle{Optica}
\bvolume{11}(\bissue{12}),
\bfpage{1762}
(\byear{2024})
\doiurl{10.1364/OPTICA.537655}
\end{barticle}
\endbibitem

\bibitem[\protect\citeauthoryear{Liu et~al.}{2024b}]{Liu2024a}
\begin{barticle}
\bauthor{\bsnm{Liu}, \binits{Y.-C.}},
\bauthor{\bsnm{Cheng}, \binits{Y.-B.}},
\bauthor{\bsnm{Pan}, \binits{X.-B.}},
\bauthor{\bsnm{Sun}, \binits{Z.-Z.}},
\bauthor{\bsnm{Pan}, \binits{D.}},
\bauthor{\bsnm{Long}, \binits{G.-L.}}:
\batitle{{Quantum integrated sensing and communication via entanglement}}.
\bjtitle{Phys. Rev. Appl.}
\bvolume{22}(\bissue{3}),
\bfpage{034051}
(\byear{2024})
\doiurl{10.1103/PhysRevApplied.22.034051}
\end{barticle}
\endbibitem

\bibitem[\protect\citeauthoryear{Dalvit et~al.}{2024}]{Dalvit2024}
\begin{barticle}
\bauthor{\bsnm{Dalvit}, \binits{D.A.R.}},
\bauthor{\bsnm{Volkoff}, \binits{T.J.}},
\bauthor{\bsnm{Choi}, \binits{Y.-S.}},
\bauthor{\bsnm{Azad}, \binits{A.K.}},
\bauthor{\bsnm{Chen}, \binits{H.-T.}},
\bauthor{\bsnm{Milonni}, \binits{P.W.}}:
\batitle{{Quantum Frequency Combs with Path Identity for Quantum Remote
  Sensing}}.
\bjtitle{Phys. Rev. X}
\bvolume{14}(\bissue{4}),
\bfpage{041058}
(\byear{2024})
\doiurl{10.1103/PhysRevX.14.041058}
\end{barticle}
\endbibitem

\bibitem[\protect\citeauthoryear{Malia et~al.}{2022}]{Malia2022a}
\begin{barticle}
\bauthor{\bsnm{Malia}, \binits{B.K.}},
\bauthor{\bsnm{Wu}, \binits{Y.}},
\bauthor{\bsnm{Mart{\'{i}}nez-Rinc{\'{o}}n}, \binits{J.}},
\bauthor{\bsnm{Kasevich}, \binits{M.A.}}:
\batitle{{Distributed quantum sensing with mode-entangled spin-squeezed atomic
  states}}.
\bjtitle{Nature}
\bvolume{612}(\bissue{7941}),
\bfpage{661}--\blpage{665}
(\byear{2022})
\doiurl{10.1038/s41586-022-05363-z}
\end{barticle}
\endbibitem

\bibitem[\protect\citeauthoryear{Zhao et~al.}{2021}]{Zhao2021d}
\begin{barticle}
\bauthor{\bsnm{Zhao}, \binits{S.-R.}},
\bauthor{\bsnm{Zhang}, \binits{Y.-Z.}},
\bauthor{\bsnm{Liu}, \binits{W.-Z.}},
\bauthor{\bsnm{Guan}, \binits{J.-Y.}},
\bauthor{\bsnm{Zhang}, \binits{W.}},
\bauthor{\bsnm{Li}, \binits{C.-L.}},
\bauthor{\bsnm{Bai}, \binits{B.}},
\bauthor{\bsnm{Li}, \binits{M.-H.}},
\bauthor{\bsnm{Liu}, \binits{Y.}},
\bauthor{\bsnm{You}, \binits{L.}},
\bauthor{\bsnm{Zhang}, \binits{J.}},
\bauthor{\bsnm{Fan}, \binits{J.}},
\bauthor{\bsnm{Xu}, \binits{F.}},
\bauthor{\bsnm{Zhang}, \binits{Q.}},
\bauthor{\bsnm{Pan}, \binits{J.-W.}}:
\batitle{{Field Demonstration of Distributed Quantum Sensing without
  Post-Selection}}.
\bjtitle{Phys. Rev. X}
\bvolume{11}(\bissue{3}),
\bfpage{031009}
(\byear{2021})
\doiurl{10.1103/PhysRevX.11.031009}
\end{barticle}
\endbibitem

\bibitem[\protect\citeauthoryear{Guo et~al.}{2020}]{Guo2020}
\begin{barticle}
\bauthor{\bsnm{Guo}, \binits{X.}},
\bauthor{\bsnm{Breum}, \binits{C.R.}},
\bauthor{\bsnm{Borregaard}, \binits{J.}},
\bauthor{\bsnm{Izumi}, \binits{S.}},
\bauthor{\bsnm{Larsen}, \binits{M.V.}},
\bauthor{\bsnm{Gehring}, \binits{T.}},
\bauthor{\bsnm{Christandl}, \binits{M.}},
\bauthor{\bsnm{Neergaard-Nielsen}, \binits{J.S.}},
\bauthor{\bsnm{Andersen}, \binits{U.L.}}:
\batitle{{Distributed quantum sensing in a continuous-variable entangled
  network}}.
\bjtitle{Nat. Phys.}
\bvolume{16}(\bissue{3}),
\bfpage{281}--\blpage{284}
(\byear{2020})
\doiurl{10.1038/s41567-019-0743-x}
\end{barticle}
\endbibitem

\bibitem[\protect\citeauthoryear{Gessner et~al.}{2020}]{Gessner2020}
\begin{barticle}
\bauthor{\bsnm{Gessner}, \binits{M.}},
\bauthor{\bsnm{Smerzi}, \binits{A.}},
\bauthor{\bsnm{Pezz{\`{e}}}, \binits{L.}}:
\batitle{{Multiparameter squeezing for optimal quantum enhancements in sensor
  networks}}.
\bjtitle{Nat. Commun.}
\bvolume{11}(\bissue{1}),
\bfpage{3817}
(\byear{2020})
\doiurl{10.1038/s41467-020-17471-3}
\end{barticle}
\endbibitem

\bibitem[\protect\citeauthoryear{Proctor et~al.}{2018}]{Proctor2018}
\begin{barticle}
\bauthor{\bsnm{Proctor}, \binits{T.J.}},
\bauthor{\bsnm{Knott}, \binits{P.A.}},
\bauthor{\bsnm{Dunningham}, \binits{J.A.}}:
\batitle{{Multiparameter Estimation in Networked Quantum Sensors}}.
\bjtitle{Phys. Rev. Lett.}
\bvolume{120}(\bissue{8}),
\bfpage{080501}
(\byear{2018})
\doiurl{10.1103/PhysRevLett.120.080501}
\end{barticle}
\endbibitem

\bibitem[\protect\citeauthoryear{Hassani et~al.}{2025}]{Hassani2024}
\begin{barticle}
\bauthor{\bsnm{Hassani}, \binits{M.}},
\bauthor{\bsnm{Scheiner}, \binits{S.}},
\bauthor{\bsnm{Paris}, \binits{M.G.A.}},
\bauthor{\bsnm{Markham}, \binits{D.}}:
\batitle{{Privacy in Networks of Quantum Sensors}}.
\bjtitle{Phys. Rev. Lett.}
\bvolume{134}(\bissue{3}),
\bfpage{030802}
(\byear{2025})
\doiurl{10.1103/PhysRevLett.134.030802}
{\href{https://arxiv.org/abs/2408.01711}{{arXiv:2408.01711}}}
\end{barticle}
\endbibitem

\bibitem[\protect\citeauthoryear{Alushi and {Di Candia}}{2026}]{Alushi2025}
\begin{barticle}
\bauthor{\bsnm{Alushi}, \binits{U.}},
\bauthor{\bsnm{{Di Candia}}, \binits{R.}}:
\batitle{{Privacy in distributed quantum sensing with Gaussian quantum
  networks}}.
\bjtitle{npj Quantum Inf.}
(\byear{2026})
\doiurl{10.1038/s41534-026-01266-3}
{\href{https://arxiv.org/abs/2509.22103}{{arXiv:2509.22103}}}
\end{barticle}
\endbibitem

\bibitem[\protect\citeauthoryear{Rosati}{2024}]{Rosati2022}
\begin{barticle}
\bauthor{\bsnm{Rosati}, \binits{M.}}:
\batitle{{A learning theory for quantum photonic processors and beyond}}.
\bjtitle{Quantum}
\bvolume{8},
\bfpage{1433}
(\byear{2024})
\doiurl{10.22331/q-2024-08-08-1433}
{\href{https://arxiv.org/abs/2209.03075}{{arXiv:2209.03075}}}
\end{barticle}
\endbibitem

\bibitem[\protect\citeauthoryear{Rosati}{2021}]{Rosati2021}
\begin{bchapter}
\bauthor{\bsnm{Rosati}, \binits{M.}}:
\bctitle{{Performance of Coherent Frequency-Shift Keying for Classical
  Communication on Quantum Channels}}.
In: \bbtitle{2021 IEEE Int. Symp. Inf. Theory},
pp. \bfpage{902}--\blpage{905}.
\bpublisher{IEEE}, \blocation{???}
(\byear{2021}).
\doiurl{10.1109/ISIT45174.2021.9517959} .
\burl{https://ieeexplore.ieee.org/document/9517959/}
\end{bchapter}
\endbibitem

\bibitem[\protect\citeauthoryear{Barz et~al.}{2013}]{Barz2013}
\begin{barticle}
\bauthor{\bsnm{Barz}, \binits{S.}},
\bauthor{\bsnm{Fitzsimons}, \binits{J.F.}},
\bauthor{\bsnm{Kashefi}, \binits{E.}},
\bauthor{\bsnm{Walther}, \binits{P.}}:
\batitle{{Experimental verification of quantum computation}}.
\bjtitle{Nat. Phys.}
\bvolume{9}(\bissue{11}),
\bfpage{727}--\blpage{731}
(\byear{2013})
\doiurl{10.1038/nphys2763}
\end{barticle}
\endbibitem

\bibitem[\protect\citeauthoryear{Polacchi et~al.}{2023}]{Polacchi2023}
\begin{barticle}
\bauthor{\bsnm{Polacchi}, \binits{B.}},
\bauthor{\bsnm{Leichtle}, \binits{D.}},
\bauthor{\bsnm{Limongi}, \binits{L.}},
\bauthor{\bsnm{Carvacho}, \binits{G.}},
\bauthor{\bsnm{Milani}, \binits{G.}},
\bauthor{\bsnm{Spagnolo}, \binits{N.}},
\bauthor{\bsnm{Kaplan}, \binits{M.}},
\bauthor{\bsnm{Sciarrino}, \binits{F.}},
\bauthor{\bsnm{Kashefi}, \binits{E.}}:
\batitle{{Multi-client distributed blind quantum computation with the Qline
  architecture}}.
\bjtitle{Nat. Commun.}
\bvolume{14}(\bissue{1}),
\bfpage{7743}
(\byear{2023})
\doiurl{10.1038/s41467-023-43617-0}
\end{barticle}
\endbibitem

\bibitem[\protect\citeauthoryear{Wei et~al.}{2025}]{Wei2025}
\begin{barticle}
\bauthor{\bsnm{Wei}, \binits{Y.-C.}},
\bauthor{\bsnm{Stas}, \binits{P.-J.}},
\bauthor{\bsnm{Suleymanzade}, \binits{A.}},
\bauthor{\bsnm{Baranes}, \binits{G.}},
\bauthor{\bsnm{Machado}, \binits{F.}},
\bauthor{\bsnm{Huan}, \binits{Y.Q.}},
\bauthor{\bsnm{Knaut}, \binits{C.M.}},
\bauthor{\bsnm{Ding}, \binits{S.W.}},
\bauthor{\bsnm{Merz}, \binits{M.}},
\bauthor{\bsnm{Knall}, \binits{E.N.}},
\bauthor{\bsnm{Yazlar}, \binits{U.}},
\bauthor{\bsnm{Sirotin}, \binits{M.}},
\bauthor{\bsnm{Wang}, \binits{I.W.}},
\bauthor{\bsnm{Machielse}, \binits{B.}},
\bauthor{\bsnm{Yelin}, \binits{S.F.}},
\bauthor{\bsnm{Borregaard}, \binits{J.}},
\bauthor{\bsnm{Park}, \binits{H.}},
\bauthor{\bsnm{Lon{\v{c}}ar}, \binits{M.}},
\bauthor{\bsnm{Lukin}, \binits{M.D.}}:
\batitle{{Universal distributed blind quantum computing with solid-state
  qubits}}.
\bjtitle{Science (80-. ).}
\bvolume{388}(\bissue{6746}),
\bfpage{509}--\blpage{513}
(\byear{2025})
\doiurl{10.1126/science.adu6894}
\end{barticle}
\endbibitem

\bibitem[\protect\citeauthoryear{Dong et~al.}{2023}]{Dong2023}
\begin{barticle}
\bauthor{\bsnm{Dong}, \binits{Y.}},
\bauthor{\bsnm{Liu}, \binits{F.}},
\bauthor{\bsnm{Xiong}, \binits{Y.}}:
\batitle{{Joint Receiver Design for Integrated Sensing and Communications}}.
\bjtitle{IEEE Commun. Lett.}
\bvolume{27}(\bissue{7}),
\bfpage{1854}--\blpage{1858}
(\byear{2023})
\doiurl{10.1109/LCOMM.2023.3274295}
{\href{https://arxiv.org/abs/2211.05535}{{arXiv:2211.05535}}}
\end{barticle}
\endbibitem

\bibitem[\protect\citeauthoryear{Chen et~al.}{2025}]{Chen2025a}
\begin{barticle}
\bauthor{\bsnm{Chen}, \binits{X.Q.}},
\bauthor{\bsnm{Zhang}, \binits{L.}},
\bauthor{\bsnm{Zheng}, \binits{Y.N.}},
\bauthor{\bsnm{Liu}, \binits{S.}},
\bauthor{\bsnm{Huang}, \binits{Z.R.}},
\bauthor{\bsnm{Liang}, \binits{J.C.}},
\bauthor{\bsnm{{Di Renzo}}, \binits{M.}},
\bauthor{\bsnm{Galdi}, \binits{V.}},
\bauthor{\bsnm{Cui}, \binits{T.J.}}:
\batitle{{Integrated sensing and communication based on space-time-coding
  metasurfaces}}.
\bjtitle{Nat. Commun.}
\bvolume{16}(\bissue{1}),
\bfpage{1836}
(\byear{2025})
\doiurl{10.1038/s41467-025-57137-6}
\end{barticle}
\endbibitem

\bibitem[\protect\citeauthoryear{Tao et~al.}{2026}]{Tao2026}
\begin{barticle}
\bauthor{\bsnm{Tao}, \binits{Y.}},
\bauthor{\bsnm{Feng}, \binits{H.}},
\bauthor{\bsnm{Fang}, \binits{Y.}},
\bauthor{\bsnm{Xie}, \binits{X.}},
\bauthor{\bsnm{Zeng}, \binits{Y.}},
\bauthor{\bsnm{Wu}, \binits{Y.}},
\bauthor{\bsnm{Ge}, \binits{T.}},
\bauthor{\bsnm{Zhang}, \binits{Y.}},
\bauthor{\bsnm{Chen}, \binits{Z.}},
\bauthor{\bsnm{Tao}, \binits{Z.}},
\bauthor{\bsnm{Xu}, \binits{J.}},
\bauthor{\bsnm{Shu}, \binits{H.}},
\bauthor{\bsnm{Wang}, \binits{X.}},
\bauthor{\bsnm{Yu}, \binits{X.}},
\bauthor{\bsnm{Wang}, \binits{C.}}:
\batitle{{Integrated photonic ultrawideband real-time spectrum sensing for 6G
  wireless networks}}.
\bjtitle{Nat. Commun.}
\bvolume{17}(\bissue{1}),
\bfpage{3666}
(\byear{2026})
\doiurl{10.1038/s41467-026-70389-0}
\end{barticle}
\endbibitem

\bibitem[\protect\citeauthoryear{Stas et~al.}{2026}]{Stas2026}
\begin{barticle}
\bauthor{\bsnm{Stas}, \binits{P.-J.}},
\bauthor{\bsnm{Wei}, \binits{Y.-C.}},
\bauthor{\bsnm{Sirotin}, \binits{M.}},
\bauthor{\bsnm{Huan}, \binits{Y.Q.}},
\bauthor{\bsnm{Yazlar}, \binits{U.}},
\bauthor{\bsnm{{Abdo Arias}}, \binits{F.}},
\bauthor{\bsnm{Knyazev}, \binits{E.}},
\bauthor{\bsnm{Baranes}, \binits{G.}},
\bauthor{\bsnm{Machielse}, \binits{B.}},
\bauthor{\bsnm{Grandi}, \binits{S.}},
\bauthor{\bsnm{Riedel}, \binits{D.}},
\bauthor{\bsnm{Borregaard}, \binits{J.}},
\bauthor{\bsnm{Park}, \binits{H.}},
\bauthor{\bsnm{Lon{\v{c}}ar}, \binits{M.}},
\bauthor{\bsnm{Suleymanzade}, \binits{A.}},
\bauthor{\bsnm{Lukin}, \binits{M.D.}}:
\batitle{{Entanglement-assisted non-local optical interferometry in a quantum
  network}}.
\bjtitle{Nature}
\bvolume{651}(\bissue{8105}),
\bfpage{326}--\blpage{332}
(\byear{2026})
\doiurl{10.1038/s41586-026-10171-w}
\end{barticle}
\endbibitem

\bibitem[\protect\citeauthoryear{Novikov et~al.}{2025}]{Novikov2025}
\begin{barticle}
\bauthor{\bsnm{Novikov}, \binits{V.}},
\bauthor{\bsnm{Jia}, \binits{J.}},
\bauthor{\bsnm{Brasil}, \binits{T.B.}},
\bauthor{\bsnm{Grimaldi}, \binits{A.}},
\bauthor{\bsnm{Bocoum}, \binits{M.}},
\bauthor{\bsnm{Balabas}, \binits{M.}},
\bauthor{\bsnm{M{\"{u}}ller}, \binits{J.H.}},
\bauthor{\bsnm{Zeuthen}, \binits{E.}},
\bauthor{\bsnm{Polzik}, \binits{E.S.}}:
\batitle{{Hybrid quantum network for sensing in the acoustic frequency range}}.
\bjtitle{Nature}
\bvolume{643}(\bissue{8073}),
\bfpage{955}--\blpage{960}
(\byear{2025})
\doiurl{10.1038/s41586-025-09224-3}
\end{barticle}
\endbibitem

\bibitem[\protect\citeauthoryear{Crawford et~al.}{2025}]{Crawford2025}
\begin{barticle}
\bauthor{\bsnm{Crawford}, \binits{S.E.}},
\bauthor{\bsnm{Lander}, \binits{G.R.}},
\bauthor{\bsnm{Paudel}, \binits{H.P.}},
\bauthor{\bsnm{Slot}, \binits{M.R.}},
\bauthor{\bsnm{Lalam}, \binits{N.}},
\bauthor{\bsnm{Wuenschell}, \binits{J.}},
\bauthor{\bsnm{Pingree}, \binits{R.}},
\bauthor{\bsnm{Oueid}, \binits{R.}},
\bauthor{\bsnm{Wright}, \binits{R.}},
\bauthor{\bsnm{Buric}, \binits{M.}},
\bauthor{\bsnm{Brister}, \binits{M.M.}},
\bauthor{\bsnm{Duan}, \binits{Y.}}:
\batitle{{Quantum sensing for emerging energy technologies}}.
\bjtitle{Nat. Rev. Clean Technol.}
\bvolume{1}(\bissue{12}),
\bfpage{861}--\blpage{876}
(\byear{2025})
\doiurl{10.1038/s44359-025-00112-7}
\end{barticle}
\endbibitem

\bibitem[\protect\citeauthoryear{Aslam et~al.}{2023}]{Aslam2023}
\begin{barticle}
\bauthor{\bsnm{Aslam}, \binits{N.}},
\bauthor{\bsnm{Zhou}, \binits{H.}},
\bauthor{\bsnm{Urbach}, \binits{E.K.}},
\bauthor{\bsnm{Turner}, \binits{M.J.}},
\bauthor{\bsnm{Walsworth}, \binits{R.L.}},
\bauthor{\bsnm{Lukin}, \binits{M.D.}},
\bauthor{\bsnm{Park}, \binits{H.}}:
\batitle{{Quantum sensors for biomedical applications}}.
\bjtitle{Nat. Rev. Phys.}
\bvolume{5}(\bissue{3}),
\bfpage{157}--\blpage{169}
(\byear{2023})
\doiurl{10.1038/s42254-023-00558-3}
\end{barticle}
\endbibitem

\bibitem[\protect\citeauthoryear{Stray et~al.}{2022}]{Stray2022}
\begin{barticle}
\bauthor{\bsnm{Stray}, \binits{B.}},
\bauthor{\bsnm{Lamb}, \binits{A.}},
\bauthor{\bsnm{Kaushik}, \binits{A.}},
\bauthor{\bsnm{Vovrosh}, \binits{J.}},
\bauthor{\bsnm{Rodgers}, \binits{A.}},
\bauthor{\bsnm{Winch}, \binits{J.}},
\bauthor{\bsnm{Hayati}, \binits{F.}},
\bauthor{\bsnm{Boddice}, \binits{D.}},
\bauthor{\bsnm{Stabrawa}, \binits{A.}},
\bauthor{\bsnm{Niggebaum}, \binits{A.}},
\bauthor{\bsnm{Langlois}, \binits{M.}},
\bauthor{\bsnm{Lien}, \binits{Y.-H.}},
\bauthor{\bsnm{Lellouch}, \binits{S.}},
\bauthor{\bsnm{Roshanmanesh}, \binits{S.}},
\bauthor{\bsnm{Ridley}, \binits{K.}},
\bauthor{\bsnm{Villiers}, \binits{G.}},
\bauthor{\bsnm{Brown}, \binits{G.}},
\bauthor{\bsnm{Cross}, \binits{T.}},
\bauthor{\bsnm{Tuckwell}, \binits{G.}},
\bauthor{\bsnm{Faramarzi}, \binits{A.}},
\bauthor{\bsnm{Metje}, \binits{N.}},
\bauthor{\bsnm{Bongs}, \binits{K.}},
\bauthor{\bsnm{Holynski}, \binits{M.}}:
\batitle{{Quantum sensing for gravity cartography}}.
\bjtitle{Nature}
\bvolume{602}(\bissue{7898}),
\bfpage{590}--\blpage{594}
(\byear{2022})
\doiurl{10.1038/s41586-021-04315-3}
\end{barticle}
\endbibitem

\bibitem[\protect\citeauthoryear{K{\'{o}}m{\'{a}}r et~al.}{2014}]{Komar2013}
\begin{barticle}
\bauthor{\bsnm{K{\'{o}}m{\'{a}}r}, \binits{P.}},
\bauthor{\bsnm{Kessler}, \binits{E.M.}},
\bauthor{\bsnm{Bishof}, \binits{M.}},
\bauthor{\bsnm{Jiang}, \binits{L.}},
\bauthor{\bsnm{S{\o}rensen}, \binits{A.S.}},
\bauthor{\bsnm{Ye}, \binits{J.}},
\bauthor{\bsnm{Lukin}, \binits{M.D.}}:
\batitle{{A quantum network of clocks}}.
\bjtitle{Nat. Phys.}
\bvolume{10}(\bissue{8}),
\bfpage{582}--\blpage{587}
(\byear{2014})
\doiurl{10.1038/nphys3000}
{\href{https://arxiv.org/abs/1310.6045}{{arXiv:1310.6045}}}
\end{barticle}
\endbibitem

\bibitem[\protect\citeauthoryear{Giovannetti et~al.}{2011}]{Giovannetti2011}
\begin{barticle}
\bauthor{\bsnm{Giovannetti}, \binits{V.}},
\bauthor{\bsnm{Lloyd}, \binits{S.}},
\bauthor{\bsnm{Maccone}, \binits{L.}}:
\batitle{{Advances in quantum metrology}}.
\bjtitle{Nat. Photonics}
\bvolume{5}(\bissue{4}),
\bfpage{222}--\blpage{229}
(\byear{2011})
\doiurl{10.1038/nphoton.2011.35}
\end{barticle}
\endbibitem

\bibitem[\protect\citeauthoryear{Pirandola et~al.}{2018}]{Pirandola2018a}
\begin{barticle}
\bauthor{\bsnm{Pirandola}, \binits{S.}},
\bauthor{\bsnm{Bardhan}, \binits{B.R.}},
\bauthor{\bsnm{Gehring}, \binits{T.}},
\bauthor{\bsnm{Weedbrook}, \binits{C.}},
\bauthor{\bsnm{Lloyd}, \binits{S.}}:
\batitle{{Advances in photonic quantum sensing}}.
\bjtitle{Nat. Photonics}
\bvolume{12}(\bissue{12}),
\bfpage{724}--\blpage{733}
(\byear{2018})
\doiurl{10.1038/s41566-018-0301-6}
{\href{https://arxiv.org/abs/1811.01969}{{arXiv:1811.01969}}}
\end{barticle}
\endbibitem

\bibitem[\protect\citeauthoryear{Marciniak et~al.}{2022}]{Marciniak2022}
\begin{barticle}
\bauthor{\bsnm{Marciniak}, \binits{C.D.}},
\bauthor{\bsnm{Feldker}, \binits{T.}},
\bauthor{\bsnm{Pogorelov}, \binits{I.}},
\bauthor{\bsnm{Kaubruegger}, \binits{R.}},
\bauthor{\bsnm{Vasilyev}, \binits{D.V.}},
\bauthor{\bsnm{Bijnen}, \binits{R.}},
\bauthor{\bsnm{Schindler}, \binits{P.}},
\bauthor{\bsnm{Zoller}, \binits{P.}},
\bauthor{\bsnm{Blatt}, \binits{R.}},
\bauthor{\bsnm{Monz}, \binits{T.}}:
\batitle{{Optimal metrology with programmable quantum sensors}}.
\bjtitle{Nature}
\bvolume{603}(\bissue{7902}),
\bfpage{604}--\blpage{609}
(\byear{2022})
\doiurl{10.1038/s41586-022-04435-4}
\end{barticle}
\endbibitem

\bibitem[\protect\citeauthoryear{Portmann and Renner}{2022}]{Portmann2022}
\begin{barticle}
\bauthor{\bsnm{Portmann}, \binits{C.}},
\bauthor{\bsnm{Renner}, \binits{R.}}:
\batitle{{Security in quantum cryptography}}.
\bjtitle{Rev. Mod. Phys.}
\bvolume{94}(\bissue{2}),
\bfpage{025008}
(\byear{2022})
\doiurl{10.1103/RevModPhys.94.025008}
\end{barticle}
\endbibitem

\bibitem[\protect\citeauthoryear{Pirandola et~al.}{2020}]{Pirandola2020}
\begin{barticle}
\bauthor{\bsnm{Pirandola}, \binits{S.}},
\bauthor{\bsnm{Andersen}, \binits{U.L.}},
\bauthor{\bsnm{Banchi}, \binits{L.}},
\bauthor{\bsnm{Berta}, \binits{M.}},
\bauthor{\bsnm{Bunandar}, \binits{D.}},
\bauthor{\bsnm{Colbeck}, \binits{R.}},
\bauthor{\bsnm{Englund}, \binits{D.}},
\bauthor{\bsnm{Gehring}, \binits{T.}},
\bauthor{\bsnm{Lupo}, \binits{C.}},
\bauthor{\bsnm{Ottaviani}, \binits{C.}},
\bauthor{\bsnm{Pereira}, \binits{J.L.}},
\bauthor{\bsnm{Razavi}, \binits{M.}},
\bauthor{\bsnm{{Shamsul Shaari}}, \binits{J.}},
\bauthor{\bsnm{Tomamichel}, \binits{M.}},
\bauthor{\bsnm{Usenko}, \binits{V.C.}},
\bauthor{\bsnm{Vallone}, \binits{G.}},
\bauthor{\bsnm{Villoresi}, \binits{P.}},
\bauthor{\bsnm{Wallden}, \binits{P.}}:
\batitle{{Advances in quantum cryptography}}.
\bjtitle{Adv. Opt. Photonics}
\bvolume{12}(\bissue{4}),
\bfpage{1012}
(\byear{2020})
\doiurl{10.1364/AOP.361502}
{\href{https://arxiv.org/abs/1906.01645}{{arXiv:1906.01645}}}
\end{barticle}
\endbibitem

\bibitem[\protect\citeauthoryear{Renner}{2007}]{Renner2007}
\begin{barticle}
\bauthor{\bsnm{Renner}, \binits{R.}}:
\batitle{{Symmetry of large physical systems implies independence of
  subsystems}}.
\bjtitle{Nat. Phys.}
\bvolume{3}(\bissue{9}),
\bfpage{645}--\blpage{649}
(\byear{2007})
\doiurl{10.1038/nphys684}
{\href{https://arxiv.org/abs/0703069}{{arXiv:0703069}}}
{[quant-ph]}
\end{barticle}
\endbibitem

\bibitem[\protect\citeauthoryear{Yin et~al.}{2020}]{Yin2020a}
\begin{barticle}
\bauthor{\bsnm{Yin}, \binits{J.}},
\bauthor{\bsnm{Li}, \binits{Y.-H.}},
\bauthor{\bsnm{Liao}, \binits{S.-K.}},
\bauthor{\bsnm{Yang}, \binits{M.}},
\bauthor{\bsnm{Cao}, \binits{Y.}},
\bauthor{\bsnm{Zhang}, \binits{L.}},
\bauthor{\bsnm{Ren}, \binits{J.-G.}},
\bauthor{\bsnm{Cai}, \binits{W.-Q.}},
\bauthor{\bsnm{Liu}, \binits{W.-Y.}},
\bauthor{\bsnm{Li}, \binits{S.-L.}},
\bauthor{\bsnm{Shu}, \binits{R.}},
\bauthor{\bsnm{Huang}, \binits{Y.-M.}},
\bauthor{\bsnm{Deng}, \binits{L.}},
\bauthor{\bsnm{Li}, \binits{L.}},
\bauthor{\bsnm{Zhang}, \binits{Q.}},
\bauthor{\bsnm{Liu}, \binits{N.-L.}},
\bauthor{\bsnm{Chen}, \binits{Y.-A.}},
\bauthor{\bsnm{Lu}, \binits{C.-Y.}},
\bauthor{\bsnm{Wang}, \binits{X.-B.}},
\bauthor{\bsnm{Xu}, \binits{F.}},
\bauthor{\bsnm{Wang}, \binits{J.-Y.}},
\bauthor{\bsnm{Peng}, \binits{C.-Z.}},
\bauthor{\bsnm{Ekert}, \binits{A.K.}},
\bauthor{\bsnm{Pan}, \binits{J.-W.}}:
\batitle{{Entanglement-based secure quantum cryptography over 1,120
  kilometres}}.
\bjtitle{Nature}
\bvolume{582}(\bissue{7813}),
\bfpage{501}--\blpage{505}
(\byear{2020})
\doiurl{10.1038/s41586-020-2401-y}
\end{barticle}
\endbibitem

\bibitem[\protect\citeauthoryear{Para{\"{i}}so et~al.}{2021}]{Paraiso2021a}
\begin{barticle}
\bauthor{\bsnm{Para{\"{i}}so}, \binits{T.K.}},
\bauthor{\bsnm{Roger}, \binits{T.}},
\bauthor{\bsnm{Marangon}, \binits{D.G.}},
\bauthor{\bsnm{{De Marco}}, \binits{I.}},
\bauthor{\bsnm{Sanzaro}, \binits{M.}},
\bauthor{\bsnm{Woodward}, \binits{R.I.}},
\bauthor{\bsnm{Dynes}, \binits{J.F.}},
\bauthor{\bsnm{Yuan}, \binits{Z.}},
\bauthor{\bsnm{Shields}, \binits{A.J.}}:
\batitle{{A photonic integrated quantum secure communication system}}.
\bjtitle{Nat. Photonics}
\bvolume{15}(\bissue{11}),
\bfpage{850}--\blpage{856}
(\byear{2021})
\doiurl{10.1038/s41566-021-00873-0}
\end{barticle}
\endbibitem

\bibitem[\protect\citeauthoryear{Nadlinger et~al.}{2022}]{Nadlinger2022}
\begin{barticle}
\bauthor{\bsnm{Nadlinger}, \binits{D.P.}},
\bauthor{\bsnm{Drmota}, \binits{P.}},
\bauthor{\bsnm{Nichol}, \binits{B.C.}},
\bauthor{\bsnm{Araneda}, \binits{G.}},
\bauthor{\bsnm{Main}, \binits{D.}},
\bauthor{\bsnm{Srinivas}, \binits{R.}},
\bauthor{\bsnm{Lucas}, \binits{D.M.}},
\bauthor{\bsnm{Ballance}, \binits{C.J.}},
\bauthor{\bsnm{Ivanov}, \binits{K.}},
\bauthor{\bsnm{Tan}, \binits{E.Y.-Z.}},
\bauthor{\bsnm{Sekatski}, \binits{P.}},
\bauthor{\bsnm{Urbanke}, \binits{R.L.}},
\bauthor{\bsnm{Renner}, \binits{R.}},
\bauthor{\bsnm{Sangouard}, \binits{N.}},
\bauthor{\bsnm{Bancal}, \binits{J.-D.}}:
\batitle{{Experimental quantum key distribution certified by Bell's theorem}}.
\bjtitle{Nature}
\bvolume{607}(\bissue{7920}),
\bfpage{682}--\blpage{686}
(\byear{2022})
\doiurl{10.1038/s41586-022-04941-5}
\end{barticle}
\endbibitem

\bibitem[\protect\citeauthoryear{Kianvash et~al.}{2026}]{Kianvash2025a}
\begin{barticle}
\bauthor{\bsnm{Kianvash}, \binits{F.}},
\bauthor{\bsnm{Barbieri}, \binits{M.}},
\bauthor{\bsnm{Rosati}, \binits{M.}}:
\batitle{{Private Remote Phase Estimation over a Lossy Quantum Channel}}.
\bjtitle{Phys. Rev. Lett.}
\bvolume{136}(\bissue{6}),
\bfpage{060805}
(\byear{2026})
\doiurl{10.1103/rhgw-t21z}
{\href{https://arxiv.org/abs/2511.09123}{{arXiv:2511.09123}}}
\end{barticle}
\endbibitem

\bibitem[\protect\citeauthoryear{Moore and Dunningham}{2023}]{Moore2023}
\begin{botherref}
\oauthor{\bsnm{Moore}, \binits{S.W.}},
\oauthor{\bsnm{Dunningham}, \binits{J.A.}}:
{Secure quantum remote sensing without entanglement}.
AVS Quantum Sci.
\textbf{5}(1)
(2023)
\doiurl{10.1116/5.0137260}
{\href{https://arxiv.org/abs/2302.03617}{{arXiv:2302.03617}}}
\end{botherref}
\endbibitem

\bibitem[\protect\citeauthoryear{Bizzarri et~al.}{2026}]{Bizzarri2025a}
\begin{barticle}
\bauthor{\bsnm{Bizzarri}, \binits{G.}},
\bauthor{\bsnm{Barbieri}, \binits{M.}},
\bauthor{\bsnm{Manrique}, \binits{M.}},
\bauthor{\bsnm{Parisi}, \binits{M.}},
\bauthor{\bsnm{Bruni}, \binits{F.}},
\bauthor{\bsnm{Gianani}, \binits{I.}},
\bauthor{\bsnm{Rosati}, \binits{M.}}:
\batitle{{Faithful and secure distributed quantum sensing under
  general-coherent attacks}}.
\bjtitle{npj Quantum Inf.}
\bvolume{12}(\bissue{1}),
\bfpage{33}
(\byear{2026})
\doiurl{10.1038/s41534-025-01180-0}
{\href{https://arxiv.org/abs/2505.02620}{{arXiv:2505.02620}}}
\end{barticle}
\endbibitem

\bibitem[\protect\citeauthoryear{Moore and Dunningham}{2025}]{Moore2025}
\begin{barticle}
\bauthor{\bsnm{Moore}, \binits{S.W.}},
\bauthor{\bsnm{Dunningham}, \binits{J.A.}}:
\batitle{{Secure quantum-enhanced measurements on a network of sensors}}.
\bjtitle{Phys. Rev. A}
\bvolume{111}(\bissue{1}),
\bfpage{012616}
(\byear{2025})
\doiurl{10.1103/PhysRevA.111.012616}
\end{barticle}
\endbibitem

\bibitem[\protect\citeauthoryear{Ho et~al.}{2025}]{Ho2024}
\begin{botherref}
\oauthor{\bsnm{Ho}, \binits{J.}},
\oauthor{\bsnm{Webb}, \binits{J.W.}},
\oauthor{\bsnm{Brooks}, \binits{R.M.J.}},
\oauthor{\bsnm{Grasselli}, \binits{F.}},
\oauthor{\bsnm{Gauger}, \binits{E.}},
\oauthor{\bsnm{Fedrizzi}, \binits{A.}}:
{Quantum-private distributed sensing}
(2025)
{\href{https://arxiv.org/abs/2410.00970}{{arXiv:2410.00970}}}
\end{botherref}
\endbibitem

\bibitem[\protect\citeauthoryear{Huang et~al.}{2019}]{Huang2019a}
\begin{barticle}
\bauthor{\bsnm{Huang}, \binits{Z.}},
\bauthor{\bsnm{Macchiavello}, \binits{C.}},
\bauthor{\bsnm{Maccone}, \binits{L.}}:
\batitle{{Cryptographic quantum metrology}}.
\bjtitle{Phys. Rev. A}
\bvolume{99}(\bissue{2}),
\bfpage{022314}
(\byear{2019})
\doiurl{10.1103/PhysRevA.99.022314}
\end{barticle}
\endbibitem

\bibitem[\protect\citeauthoryear{Yin et~al.}{2019}]{Yin2019}
\begin{botherref}
\oauthor{\bsnm{Yin}, \binits{P.}},
\oauthor{\bsnm{Takeuchi}, \binits{Y.}},
\oauthor{\bsnm{Zhang}, \binits{W.-H.}},
\oauthor{\bsnm{Yin}, \binits{Z.-Q.}},
\oauthor{\bsnm{Matsuzaki}, \binits{Y.}},
\oauthor{\bsnm{Peng}, \binits{X.-X.}},
\oauthor{\bsnm{Xu}, \binits{X.-Y.}},
\oauthor{\bsnm{Xu}, \binits{J.-S.}},
\oauthor{\bsnm{Tang}, \binits{J.-S.}},
\oauthor{\bsnm{Zhou}, \binits{Z.-Q.}},
\oauthor{\bsnm{Chen}, \binits{G.}},
\oauthor{\bsnm{Li}, \binits{C.-F.}},
\oauthor{\bsnm{Guo}, \binits{G.-C.}}:
{Experimental demonstration of secure quantum remote sensing}
(2019)
\doiurl{10.1103/PhysRevApplied.14.014065}
{\href{https://arxiv.org/abs/1907.06480}{{arXiv:1907.06480}}}
\end{botherref}
\endbibitem

\bibitem[\protect\citeauthoryear{Shettell et~al.}{2022}]{Shettell2022a}
\begin{barticle}
\bauthor{\bsnm{Shettell}, \binits{N.}},
\bauthor{\bsnm{Kashefi}, \binits{E.}},
\bauthor{\bsnm{Markham}, \binits{D.}}:
\batitle{{Cryptographic approach to quantum metrology}}.
\bjtitle{Phys. Rev. A}
\bvolume{105}(\bissue{1}),
\bfpage{010401}
(\byear{2022})
\doiurl{10.1103/PhysRevA.105.L010401}
\end{barticle}
\endbibitem

\bibitem[\protect\citeauthoryear{Shor and Preskill}{2000}]{Shor2000}
\begin{barticle}
\bauthor{\bsnm{Shor}, \binits{P.W.}},
\bauthor{\bsnm{Preskill}, \binits{J.}}:
\batitle{{Simple Proof of Security of the BB84 Quantum Key Distribution
  Protocol}}.
\bjtitle{Phys. Rev. Lett.}
\bvolume{85}(\bissue{2}),
\bfpage{441}--\blpage{444}
(\byear{2000})
\doiurl{10.1103/PhysRevLett.85.441}
{\href{https://arxiv.org/abs/0003004}{{arXiv:0003004}}}
{[quant-ph]}
\end{barticle}
\endbibitem

\bibitem[\protect\citeauthoryear{Kraus et~al.}{2005}]{Kraus2005}
\begin{barticle}
\bauthor{\bsnm{Kraus}, \binits{B.}},
\bauthor{\bsnm{Gisin}, \binits{N.}},
\bauthor{\bsnm{Renner}, \binits{R.}}:
\batitle{{Lower and Upper Bounds on the Secret-Key Rate for Quantum Key
  Distribution Protocols Using One-Way Classical Communication}}.
\bjtitle{Phys. Rev. Lett.}
\bvolume{95}(\bissue{8}),
\bfpage{080501}
(\byear{2005})
\doiurl{10.1103/PhysRevLett.95.080501}
{\href{https://arxiv.org/abs/0410215}{{arXiv:0410215}}}
{[quant-ph]}
\end{barticle}
\endbibitem

\bibitem[\protect\citeauthoryear{Helstrom}{1976}]{helstromBOOK}
\begin{bbook}
\bauthor{\bsnm{Helstrom}, \binits{C.W.}}:
\bbtitle{{Quantum Detection and Estimation Theory}}
vol. \bseriesno{84},
p. \bfpage{309}.
\bpublisher{Academic press},
\blocation{New York}
(\byear{1976})
\end{bbook}
\endbibitem

\bibitem[\protect\citeauthoryear{Langford et~al.}{2005}]{Langford2005}
\begin{barticle}
\bauthor{\bsnm{Langford}, \binits{N.K.}},
\bauthor{\bsnm{Weinhold}, \binits{T.J.}},
\bauthor{\bsnm{Prevedel}, \binits{R.}},
\bauthor{\bsnm{Resch}, \binits{K.J.}},
\bauthor{\bsnm{Gilchrist}, \binits{A.}},
\bauthor{\bsnm{O'Brien}, \binits{J.L.}},
\bauthor{\bsnm{Pryde}, \binits{G.J.}},
\bauthor{\bsnm{White}, \binits{A.G.}}:
\batitle{{Demonstration of a Simple Entangling Optical Gate and Its Use in
  Bell-State Analysis}}.
\bjtitle{Phys. Rev. Lett.}
\bvolume{95}(\bissue{21}),
\bfpage{210504}
(\byear{2005})
\doiurl{10.1103/PhysRevLett.95.210504}
\end{barticle}
\endbibitem

\bibitem[\protect\citeauthoryear{Mukhopadhyay et~al.}{2025}]{Mukhopadhyay2025}
\begin{botherref}
\oauthor{\bsnm{Mukhopadhyay}, \binits{C.}},
\oauthor{\bsnm{Bayat}, \binits{A.}},
\oauthor{\bsnm{Montenegro}, \binits{V.}},
\oauthor{\bsnm{Paris}, \binits{M.G.A.}}:
{Beating joint quantum estimation limits with stepwise multiparameter
  metrology}
(2025)
{\href{https://arxiv.org/abs/2506.06075}{{arXiv:2506.06075}}}
\end{botherref}
\endbibitem

\end{thebibliography}
\end{document}